\documentclass[12pt,oneside]{amsart}

\usepackage{amsthm,amsmath,amssymb,amsfonts,amscd,mathtools,,upgreek,url,
}
\usepackage{amsaddr}

\newcommand{\cS}{{\mathcal{S}}}
\newcommand{\tdt}{D_t}

\newcommand{\bc}{y}
\newcommand{\pe}{u}
\newcommand{\vv}{v}

\newcommand{\vb}{v}
\newcommand{\wf}{w}
\newcommand{\ik}{\ell}
\newcommand{\iup}{\gamma}

\newcommand{\dd}[1]{\partial_{#1}}

\newcommand{\mme}{\mathbf{m}}

\newcommand{\ap}{\mathrm{q}}

\newcommand{\mc}{\mathrm{c}}

\newcommand{\ud}{u_\ell,\dots}
\newcommand{\ua}{{\mathrm{a}}}
\newcommand{\ub}{{\mathrm{b}}}
\newcommand{\tua}{{\tilde{\ua}}}
\newcommand{\tub}{{\tilde{\ub}}}

\newcommand{\fn}{f}
\newcommand{\mv}{v}

\newcommand{\inb}{b}
\newcommand{\zc}{z}

\newcommand{\mf}{\Phi}
\newcommand{\la}{\uplambda}

\newcommand{\ff}{\mathbf{F}}
\newcommand{\tff}{\tilde{\mathbf{F}}}

\newcommand{\fik}{\mathbb{C}}

\newcommand{\sm}{\mathrm{d}}

\newcommand{\lb}{\label}
\newcommand{\er}{\eqref}

\newcommand{\zp}{\mathbb{Z}_{\ge 0}}
\newcommand{\zsp}{\mathbb{Z}_{>0}}
\newcommand{\zz}{\mathbb{Z}}

\newcommand{\pd}{\partial}
\newcommand{\lt}{\mathbf{U}}

\newcommand{\mm}{\mathbf{M}}

\newcommand{\diunew}{N}

\newcommand{\tiltdt}{\tilde{\tdt}}

\newtheorem{theorem}{Theorem}

\theoremstyle{definition}
\newtheorem{definition}{Definition}

\newtheorem{remark}{Remark}

\begin{document}

\title[Matrix Lax pairs and Miura-type transformations for lattice equations]{Matrix Lax pairs 
under the gauge equivalence relation induced by the gauge group action
and Miura-type transformations \\ 
for lattice equations}
\date{}

\author{Sergei Igonin} 
\address{Center of Integrable Systems, P.G. Demidov Yaroslavl State University, Yaroslavl, Russia}
\email{s-igonin@yandex.ru}


\begin{abstract}
In this paper we explore interconnections of 
differential-difference matrix Lax representations (Lax pairs), 
gauge transformations, and discrete Miura-type transformations (MTs),
which belong to the main tools in the theory of integrable 
differential-difference (lattice) equations.

For a given equation, two matrix Lax representations (MLRs) 
are said to be gauge equivalent if one of them 
can be obtained from the other by means of a (local) matrix gauge transformation.
Matrix gauge transformations constitute an infinite-dimensional group called the matrix gauge group, 
which acts naturally on the set of MLRs of a given equation.
Two MLRs are gauge equivalent if and only if they belong to the same 
orbit of the matrix gauge group action.

For a wide class of MLRs of (vector) evolutionary differential-difference 
equations, we present results on the following questions:

1. When and how can one simplify a given MLR 
by matrix gauge transformations and bring the MLR 
to a form suitable for constructing MTs?

2. A MLR is called fake if it is gauge equivalent to a trivial MLR.
How to determine whether a given MLR is not fake?

Here and in a different publication (with E.~Chistov),
we apply results of the present paper to the following integrable examples:
\begin{itemize}
\item a $3$-component lattice introduced by D.~Zhang and D.~Chen 
in their work on Hamiltonian structures of evolutionary lattice 
equations~\cite{ZhangChen2002},
\item some rational $1$-component equations of order~$(-2,2)$ related 
to the Narita--Itoh--Bogoyavlensky lattice,
\item the $2$-component Boussinesq lattice related to the lattice $W_3$-algebra,
\item a $2$-component equation 
(introduced by G.~Mar\'i Beffa and Jing~Ping~Wang
in their work on Hamiltonian evolutions of polygons~\cite{BeffaWang2013})
which describes the evolution induced on invariants by an invariant evolution of planar polygons.
\end{itemize}
This allows us to construct 
new integrable equations (with new MLRs) connected by new MTs to known equations.
\end{abstract}

\maketitle

\section{Introduction and the main ideas}
\lb{secint}

In this paper we explore interconnections of 
differential-difference matrix Lax representations (Lax pairs)
under the gauge equivalence relation induced by the matrix gauge group action
and discrete Miura-type transformations,
which play essential roles in the theory of integrable differential-difference (lattice) equations.
Such equations are currently the subject of intensive study, since they
occupy a prominent place in the modern theory of integrable systems and its applications
in mathematical physics and geometry.
In particular, such equations emerge in discretizations of integrable 
partial differential equations (PDEs), in discretizations of various differential geometric constructions
and as chains associated with B\"acklund and Darboux transformations of PDEs 
(see, 
e.g.,~\cite{BeffaWang2013,HJN-book2016,kmw2013,LWY-book2022,suris2003} 
and references therein).

We denote by $\zsp$ and $\zp$ the sets of positive and nonnegative integers, respectively.
Let $\diunew\in\zsp$. Below we consider an equation for an $\diunew$-component vector-function 
$u=\big(u^1(n,t),\dots,u^{\diunew}(n,t)\big)$
of an integer variable~$n$ and a real or complex variable~$t$.
For any fixed integer~$\ik$, 
we have $u_\ik=(u^1_\ik,\dots,u^{\diunew}_\ik)$,
where for each $\iup=1,\dots,\diunew$ 
the component $u_\ik^\iup$ is a function
of $n,t$ given by the formula $u_\ik^\iup(n,t)=u^\iup(n+\ik,t)$.
That is, $u_\ik(n,t)=u(n+\ik,t)$. In particular, $u_0=u$.

Let $\ua,\ub\in\zz$ such that $\ua\le\ub$.
Consider an evolutionary differential-difference (lattice) equation of the form
\begin{gather}
\lb{sddenew}
\pd_t(u)=\ff(u_\ua,u_{\ua+1},\dots,u_\ub),
\end{gather}
where $\ff$ is an $\diunew$-component vector-function $\ff=(\ff^1,\dots,\ff^{\diunew})$.

If the right-hand side of~\eqref{sddenew} depends nontrivially on~$u_\ua,u_\ub$,
then equation~\eqref{sddenew} is said to be \emph{of order~$(\ua,\ub)$}.

The differential-difference equation \eqref{sddenew} 
is equivalent to the following infinite collection of differential equations
$$
\pd_t\big(u(n,t)\big)=\ff\big(u(n+\ua,t),u(n+\ua+1,t),\dots,u(n+\ub,t)\big),\qquad\quad n\in\zz.
$$
Using the components of $u=\big(u^1(n,t),\dots,u^{\diunew}(n,t)\big)$ 
and of $\ff=(\ff^1,\dots,\ff^{\diunew})$, one can rewrite~\eqref{sddenew} as
\begin{gather}
\lb{msdde}
\pd_t\big(u^i\big)=\ff^i(u_\ua^\gamma,u_{\ua+1}^\gamma,\dots,u_\ub^\gamma),\qquad\quad
i=1,\dots,\diunew,
\end{gather}
which implies 
\begin{gather}
\lb{uildde}
\pd_t\big(u^i_\ell\big)=\ff^i(u_{\ua+\ell}^\gamma,u_{\ua+1+\ell}^\gamma,\dots,u_{\ub+\ell}^\gamma),\qquad\quad
i=1,\dots,\diunew,\qquad \ell\in\zz.
\end{gather}

We use the formal theory of differential-difference equations, where one regards 
\begin{gather}
\lb{uldiunew}
u_\ell=(u^1_\ell,\dots,u^\diunew_\ell),\qquad\quad \ell\in\zz,
\end{gather}
as independent quantities, which are called \emph{dynamical variables}.

In this paper, the notation of the type $\fn=\fn(\ud)$ means 
that a function~$\fn$ depends on a finite 
number of the dynamical variables $u_\ell^{\gamma}$ for $\ell\in\zz$ and $\gamma=1,\dots,\diunew$.
The notation of the type $\fn=\fn(u_\alpha,\dots,u_\beta)$ 
for some integers $\alpha\le\beta$ means that $\fn$ 
may depend on $u_\ell^{\gamma}$ for $\ell=\alpha,\dots,\beta$ 
and $\gamma=1,\dots,\diunew$.

We denote by~$\cS$ the \emph{shift operator} with respect to the variable~$n$.
For any function $g=g(n,t)$ one has the function~$\cS(g)$
such that $\cS(g)(n,t)=g(n+1,t)$.
Furthermore, for each $k\in\zz$, we have the $k$th power~$\cS^k$ 
of the operator~$\cS$ and the formula $\cS^k(g)(n,t)=g(n+k,t)$.

Since $u_\ell$ corresponds to $u(n+\ell,t)$, the operator~$\cS$ 
and its powers~$\cS^k$ for $k\in\zz$ act on functions of~$u_\ell$ by means of the rules
\begin{gather}
\lb{csuf}
\cS(u_\ell)=u_{\ell+1},\qquad\cS^k(u_\ell)=u_{\ell+k},\qquad
\cS^k\big(\fn(u_\ell,\dots)\big)=\fn(\cS^k(u_{\ell}),\dots).
\end{gather}
That is, applying $\cS^k$ to a function $\fn=\fn(\ud)$, 
we replace~$u_\ell^{\gamma}$ by~$u_{\ell+k}^{\gamma}$ in~$\fn$ for all $\ell\in\zz$ 
and $\gamma=1,\dots,\diunew$.

The \emph{total derivative operator~$\tdt$ corresponding to~\eqref{sddenew}} 
acts on functions of the variables~$u_\ell=(u^1_\ell,\dots,u^\diunew_\ell)$ as follows
\begin{gather}
\lb{dtfu}
\tdt\big(\fn(\ud)\big)=\sum_{\ell,\gamma}\cS^\ell(\ff^\gamma)\cdot\frac{\pd \fn}{\pd u_\ell^\gamma},
\end{gather}
where $\ff^\gamma=\ff^\gamma(u_\ua,\dots,u_\ub)$ are the components of the vector-function 
$\ff=(\ff^1,\dots,\ff^\diunew)$ from~\eqref{sddenew}.
Formula~\er{dtfu} reflects the chain rule for the derivative with respect to~$t$, 
taking into account equations~\er{uildde}.
From~\er{csuf},~\er{dtfu} it follows that, for any function $h=h(\ud)$,
one has $\tdt\big(\cS(h)\big)=\cS\big(\tdt(h)\big)$.

In this paper, matrix-functions are sometimes called simply matrices.

\begin{remark}
\lb{rratf}
In this paper we deal with functions which may depend 
on a finite number of the dynamical variables~\er{uldiunew} and a complex parameter~$\la$.
For simplicity of exposition, we consider only functions which depend rationally on
the variables~\er{uldiunew} and~$\la$.

In fact, many results of this paper can be extended to more general classes of functions, 
but such extension requires more technical and heavy proofs.
In order to present the main ideas clearly and concisely, 
we restrict ourselves to the case of rational functions.

In particular, for each $\sm\in\zsp$, 
when we consider a $\sm\times\sm$ matrix-function
$\mm=\mm(\ud,\la)$, we assume that each entry of the matrix~$\mm$ 
is a rational function of the variables~\er{uldiunew} and~$\la$.
When we say that such a matrix~$\mm(\ud,\la)$ is assumed to be invertible, 
we mean that the determinant $\det\big(\mm(\ud,\la)\big)$ 
is not identically zero (as a rational function). 
Then one has the rational function $\big(\det\big(\mm(\ud,\la)\big)\big)^{-1}$,
which allows us to compute the inverse matrix $\mm^{-1}=\big(\mm(\ud,\la)\big)^{-1}$
as a rational matrix-function.

Consider two $\sm\times\sm$ matrix-functions $\mm_i=\mm_i(\ud,\la)$, $\,i=1,2$,
which are invertible in the above sense. Basic properties of rational functions and determinants 
imply that the inverse matrix $\mm_1^{-1}=\big(\mm_1(\ud,\la)\big)^{-1}$
and the product $\mm_1(\ud,\la)\cdot\mm_2(\ud,\la)$ are invertible as well.
\end{remark}

\begin{definition}
\lb{dmlpgtnew}
Let $\sm\in\zsp$. 
Let $\mm=\mm(\ud,\la)$ and $\lt=\lt(\ud,\la)$ be $\sm\times\sm$ matrix-functions
depending on the dynamical variables~\er{uldiunew} and a complex parameter~$\la$.
Suppose that $\mm$ is invertible (in the sense of Remark~\ref{rratf}) and one has
\begin{gather}
\lb{lr}
\tdt(\mm)=\cS(\lt)\cdot\mm-\mm\cdot\lt,
\end{gather}
where $\tdt$ is given by~\eqref{dtfu}.
Then the pair $(\mm,\,\lt)$ is called a \emph{matrix Lax representation} (MLR) for equation~\eqref{sddenew}.

The connection between~\eqref{lr} and~\eqref{sddenew} is as follows.
The components~$\ff^\iup$ of the right-hand side $\ff(u_\ua,\dots,u_\ub)$ 
of~\eqref{sddenew} appear in formula~\er{dtfu} for the operator~$\tdt$, which is used in~\er{lr}.

Relation~\er{lr} implies that the following (overdetermined) auxiliary linear system
\begin{gather}
\lb{syspsi}
\cS(\Psi)=\mm\cdot\Psi,\qquad\quad
\pd_t(\Psi)=\lt\cdot\Psi
\end{gather}
is compatible modulo equation~\eqref{sddenew}. 
Here $\Psi=\Psi(n,t)$ is an invertible $\sm\times\sm$ matrix-function.

We say that the matrix $\mm=\mm(\ud,\la)$ is the \emph{$\cS$-part} of the MLR $(\mm,\,\lt)$.

Then for any invertible $\sm\times\sm$ matrix-function  
$\mathbf{g}=\mathbf{g}(\ud,\la)$ the matrices
\begin{gather}
\lb{ggtmtu}
\hat \mm=\cS(\mathbf{g})\cdot \mm\cdot\mathbf{g}^{-1},\qquad\quad
\hat\lt=\tdt(\mathbf{g})\cdot\mathbf{g}^{-1}+
\mathbf{g}\cdot\lt\cdot\mathbf{g}^{-1}
\end{gather}
form a MLR for equation~\eqref{sddenew} as well.
Here the $\sm\times\sm$ matrix $\mathbf{g}=\mathbf{g}(\ud,\la)$
is invertible in the sense of Remark~\ref{rratf} and is called a \emph{gauge transformation}.
The MLR $(\hat \mm,\,\hat\lt)$ is \emph{gauge equivalent} 
to the MLR $(\mm,\,\lt)$, and one can say that
$(\hat \mm,\,\hat\lt)$ is obtained from $(\mm,\,\lt)$ 
by means of the gauge transformation~$\mathbf{g}$.

Such gauge transformations~$\mathbf{g}=\mathbf{g}(\ud,\la)$ 
constitute an infinite-dimensional group 
with respect to the multiplication of invertible $\sm\times\sm$ matrices, 
which is called the \emph{matrix gauge group}.
Formulas~\er{ggtmtu} determine an action of the matrix gauge group
on the set of MLRs of a given equation~\eqref{sddenew}.
Two MLRs are gauge equivalent if and only if they belong to the same 
orbit of the matrix gauge group action.
\end{definition}
\begin{remark}
In~\cite{kmw2013,BIg2016} and some other publications, 
MLRs for differential-difference equations are called \emph{Darboux--Lax representations}, 
since many of them arise from Darboux transformations 
of partial differential equations (see, e.g.,~\cite{kmw2013}).
\end{remark}

\begin{remark}
\lb{rstrict}
Definition~\ref{dmlpgtnew} essentially says that $(\mm,\,\lt)$ is a MLR
for equation~\er{sddenew} if $\mm$ is invertible and relation~\er{lr} 
turns to identity in virtue of~\er{sddenew}. As noted in Definition~\ref{dmlpgtnew},
relation~\er{lr} implies that the overdetermined auxiliary linear system~\er{syspsi}
is compatible modulo equation~\eqref{sddenew}. 
That is, the compatibility of system~\er{syspsi} holds as a consequence of~\eqref{sddenew}. 

In the existing literature,
some authors define MLRs (or Lax pairs) in a more restrictive way, including 
the additional requirement that ``\emph{the compatibility of the overdetermined auxiliary linear system~\er{syspsi}
must be equivalent to equation~\eqref{sddenew}}''.
We do not need this additional requirement, 
since the methods of the present paper are applicable in the case when 
the compatibility of system~\er{syspsi} holds as a consequence of equation~\eqref{sddenew}.
\end{remark}

\begin{definition}
\lb{dtrmlp}
We consider MLRs for a given equation~\er{sddenew}.
A MLR~$(\mm,\,\lt)$ is said to be \emph{trivial} if $\mm$ does not depend 
on~$u_\ell$ for any $\ell\in\zz$.
Then \er{dtfu} yields $\tdt(\mm)=0$, and from~\er{lr} we get 
$\cS(\lt)=\mm\cdot\lt\cdot\mm^{-1}$, which 
implies that $\lt$ does not depend on~$u_\ell$ either.
Therefore, a trivial MLR does not provide any information about equation~\er{sddenew}.

A MLR is called \emph{fake} if it is gauge equivalent to a trivial MLR.
That is, fake MLRs are obtained from trivial MLRs by applying gauge transformations, 
and this can be done independently of equation~\er{sddenew}.
Hence a fake MLR does not provide any information about equation~\er{sddenew}.
\end{definition}

According to Definition~\ref{dmlpgtnew}, 
in a MLR~$(\mm,\,\lt)$ the matrix~$\mm$ may depend on 
a finite number of the variables~$u_{\ell}=(u^1_{\ell},\dots,u^\diunew_{\ell})$, $\,\ell\in\zz$, 
and a parameter~$\la$. For any fixed integers $r_1,\dots,r_\diunew$, we can relabel 
\begin{gather}
\lb{relabnew}
u^1_\ell\,\mapsto\,u^1_{\ell+r_1},\,\quad\dots\,\quad,\,\quad 
u^{\diunew}_\ell\,\mapsto\,u^{\diunew}_{\ell+r_\diunew}\qquad\quad\forall\,\ell\in\zz.
\end{gather}
Relabeling~\er{relabnew} means that in equation~\er{sddenew} 
we make the following invertible change of variables 
\begin{gather}
\notag
u^1(n,t)\,\mapsto\,u^1(n+r_1,t),\,\quad\dots\,\quad,\,\quad 
u^{\diunew}(n,t)\,\mapsto\,u^{\diunew}(n+r_\diunew,t).
\end{gather}
Applying a relabeling~\er{relabnew} with suitable fixed $r_1,\dots,r_\diunew\in\zz$, 
one can transform~$\mm$ to the form 
$\mm=\mm(u_0,\dots,u_p,\la)$ for some $p\in\zp$.

In this paper we study MLRs $(\mm,\,\lt)$ with $\cS$-part~$\mm$
depending only on~$u_0,\,u_1,\,u_2,\,\la$. That is, we assume that $\mm$ is of the form
\begin{gather}
\lb{muu2}
\mm=\mm(u_0,u_1,u_2,\la).
\end{gather}
To our knowledge, the majority of known examples of MLRs belong to this class 
or can be transformed to the form~\er{muu2} by means of a suitable relabeling~\er{relabnew}.

In formula~\er{muu2} we do not require nontrivial dependence of~$\mm$ 
on all the variables~$u_0,\,u_1,\,u_2$.
In particular, the cases $\mm=\mm(u_0,u_1,\la)$ and $\mm=\mm(u_0,\la)$ are included in~\er{muu2}.

In what follows, for any function $w=w(n,t)$ and each $\ell\in\zz$, 
we denote by~$w_\ell$ the function $w_\ell(n,t)=w(n+\ell,t)$.
That is, $w_\ell=\cS^\ell(w)$. In particular, $w_0=w$.

Now let $\tua,\tub\in\zz$, $\,\tua\le\tub$, 
and consider another evolutionary differential-difference equation
\begin{gather}
\lb{vddenew}
\pd_t(\mv)=\tff(\mv_\tua,\mv_{\tua+1},\dots,\mv_\tub)
\end{gather}
for an $\diunew$-component vector-function $\mv=\big(\mv^1(n,t),\dots,\mv^{\diunew}(n,t)\big)$.
In the right-hand side of~\er{vddenew} one has
an $\diunew$-component vector-function $\tff=(\tff^1,\dots,\tff^{\diunew})$.

\begin{definition}
\lb{defmtt}
A \emph{Miura-type transformation} (MT) from equation~\eqref{vddenew} 
to equation~\eqref{sddenew} is determined by an expression of the form
\begin{gather}
\lb{uvf}
u=\mf(\mv_\ik,\dots)
\end{gather}
with the following requirements:
\begin{enumerate}
	\item The right-hand side of~\er{uvf} 
	is an $\diunew$-component vector-function $\mf=(\mf^1,\dots,\mf^{\diunew})$
	depending on a finite number of the dynamical variables 
	$\mv_\ik=(\mv^1_\ik,\dots,\mv^\diunew_\ik)$, $\,\ik\in\zz$.
	\item If $\mv=\mv(n,t)$ obeys equation~\eqref{vddenew} 
	then $u=u(n,t)$ determined by~\eqref{uvf} obeys equation~\eqref{sddenew}.
\end{enumerate}
More precisely, the second requirement means that we must have relations~\er{dtmfi} explained below.
The components of the vector formula~\eqref{uvf} are 
\begin{gather}
\lb{uimfi}
u^i=\mf^i(\mv^\iup_\ik,\dots),\qquad\quad i=1,\dots,\diunew.
\end{gather}
Recall that \er{sddenew} is equivalent to~\er{msdde}.
Substituting the right-hand side of~\eqref{uimfi} in place of~$u^i$ in~\eqref{msdde}, we get 
\begin{gather}
\lb{dtmfi}
\tiltdt\big(\mf^i(\mv^\iup_\ik,\dots)\big)=
\ff^i\big(\cS^\ua(\mf^\iup),\cS^{\ua+1}(\mf^\iup),\dots,\cS^\ub(\mf^\iup)\big),\qquad\quad
i=1,\dots,\diunew,
\end{gather}
which must be valid identically in the variables $\mv^\gamma_\ell$.
Here $\tiltdt$ is the total derivative operator corresponding to~\eqref{vddenew}, 
so that in the left-hand side of~\er{dtmfi} we have 
$\tiltdt\big(\mf^i(\mv^\iup_\ik,\dots)\big)=%
\sum_{\ell,\gamma}\cS^\ell\big(\tff^\gamma\big)\cdot\dfrac{\pd\mf^i}{\pd \mv^\gamma_\ell}$,
where $\tff^\gamma=\tff^\gamma(\mv_\tua,\dots,\mv_\tub)$ are the components 
of the vector-function $\tff=(\tff^1,\dots,\tff^{\diunew})$ from~\eqref{vddenew}.

\end{definition}
\begin{remark}
\lb{rtermmt}
MTs for differential-difference equations
are also called \emph{discrete substitutions}~\cite{yam1994} and 
are a discrete analog of MTs for partial differential equations~\cite{miura1968}.

MTs for partial differential equations are sometimes called \emph{differential substitutions}.
\end{remark}

In Section~\ref{srmlp} we present the following results.
\begin{itemize}
\item 
For a given MLR $(\mm,\,\lt)$ with $\cS$-part of the form~\er{muu2},
Theorem~\ref{thmuknew} provides sufficient conditions for the possibility to transform 
(by means of a gauge transformation) the MLR $(\mm,\,\lt)$ to a MLR 
$(\hat{\mm},\,\hat{\lt})$ with $\cS$-part of the form $\hat{\mm}=\hat{\mm}(u_0,\la)$. 

Namely, Theorem~\ref{thmuknew} says that, 
if $\mm=\mm(u_0,u_1,u_2,\la)$ obeys conditions~\er{pdmnew},~\er{pdm01}, 
then there is a gauge transformation $\mathbf{g}=\mathbf{g}(u_0,u_{1},\la)$ 
given by the explicit formula~\er{guu1} such that
the matrix-function $\hat{\mm}=\cS(\mathbf{g})\cdot \mm\cdot\mathbf{g}^{-1}$ 
is of the form $\hat{\mm}=\hat{\mm}(u_0,\la)$. (Thus $\hat{\mm}$ does not depend on~$u_1,u_2$.)
Hence conditions~\er{pdmnew},~\er{pdm01} allow us 
to simplify the MLR $(\mm,\,\lt)$, by applying to it some explicit gauge transformation,
eliminating the dependence on~$u_1,u_2$ in the~\mbox{$\cS$-part} of the MLR. 

Having obtained a MLR $(\hat{\mm},\,\hat{\lt})$ with $\hat{\mm}=\hat{\mm}(u_0,\la)$,
one can often use it for constructing new integrable equations with new MTs 
connecting these equations to the initial equation.
Several examples of such constructions are discussed below.

\item 
Theorem~\ref{thtrmlp} says that 
a MLR $(\mm,\,\lt)$ with $\cS$-part of the form~\er{muu2}
is gauge equivalent to a trivial MLR if and only if $\mm$ satisfies~\er{ctrmlp}.

Thus, a MLR $(\mm,\,\lt)$ is not fake (in the sense of Definition~\ref{dtrmlp})
if and only if $\mm$ does not satisfy~\er{ctrmlp}.

\item As said above, the result of Theorem~\ref{thmuknew} concerns 
MLRs $(\mm,\,\lt)$ with $\cS$-part of the form~\er{muu2}.
A similar result on MLRs with $\cS$-part of the form
$\mm=\mm(u_0,u_1,\la)$ is presented in Theorem~\ref{thmuk1}.
\end{itemize}
In the proofs of Theorems~\ref{thmuknew},~\ref{thmuk1} 
we use some ideas from~\cite{IgSimpl2024} 
on simplifications of MLRs by gauge transformations.

When one makes a classification for some class 
of integrable (partial differential, difference or differential-difference) 
equations with two independent variables,
it is often possible to obtain a few basic equations (depending on some parameters)
such that all the other equations from the considered class can be derived from 
the basic ones by means of MTs 
(see, e.g.,~\cite{mss91,yam2006,garif2017,garif2018,gram2011,meshk2008,LWY-book2022} 
and references therein).
Also, it is well known that MTs often help to obtain conservation laws 
and B\"acklund transformations for equations with two independent variables 
(see, e.g.,~\cite{miura-cl,suris2003,HJN-book2016,LWY-book2022}).

Therefore, it is desirable to develop constructions of MTs.
As discussed below, several constructions of MTs are given 
in Sections~\ref{sinb},~\ref{szhch}, using results of Section~\ref{srmlp} on~MLRs, 
and this allows us to derive 
new integrable equations (with new MLRs) connected by new MTs to known equations.

The above-mentioned Theorems~\ref{thmuknew},~\ref{thmuk1} 
from Section~\ref{srmlp} help us to obtain the following results in 
Sections~\ref{sinb},~\ref{szhch}.
\begin{itemize}
	\item We construct new $3$-component integrable 
	equations~\er{npeq},~\er{hnpeq} 
	and new MTs~\er{bu0vb}, \er{nmtpbc}, \er{hnmtpbc}, \er{cmtzc}.
	In formulas \er{npeq}, \er{hnpeq}, \er{bu0vb}, \er{nmtpbc}, \er{hnmtpbc}, \er{cmtzc}
	we have arbitrary constants~$\mc,\ap\in\fik$.
	
The functions $F^i(\mc,\ap,\vv^1_\ell,\vv^2_\ell,\vv^3_\ell,\dots)$, $\,i=1,2,3$, 
in~\er{hnpeq} are obtained by substituting~\er{hpeell} to the right-hand sides
of equations~\er{dtv1p},~\er{dtv2p},~\er{dtv3p}. We do not present 
explicit formulas for these functions, since they are rather cumbersome.
	
The MT~\er{hnmtpbc} is from~\er{hnpeq} to~\er{npeq}.
The MTs~\er{nmtpbc},~\er{cmtzc} connect equations~\er{npeq},~\er{hnpeq} 
to the $3$-component integrable equation~\er{zceq} 
introduced by D.~Zhang and D.~Chen~\cite{ZhangChen2002}.

The MT~\er{bu0vb} connects the $1$-component integrable equation~\er{dtwf} 
to the well-known integrable modified Narita--Itoh--Bogoyavlensky equation~\er{bpeq} 
related to the standard Narita--Itoh--Bogo\-yav\-len\-sky (NIB) lattice~\er{be}
by the well-known formula~\er{bmtpbc}.
The $1$-component equations~\er{be},~\er{bpeq},~\er{dtwf} are of order~$(-2,2)$.

Equation~\er{dtwf} was obtained previously 
in~\cite[equation~(161)]{BIg2016} by a different construction.
In the case when $\mc\neq 0$ and $\ap\neq 0$ equation~\er{dtwf}
can be transformed to~\cite[equation~(5.11)]{xenit2018} by scaling.

\item Also, we construct a new MLR~\er{zchmm},~\er{zchlt} 
for equation~\er{npeq} and new MLRs (described in Remarks~\ref{ridtwf},~\ref{rihnpeq}) 
for equations~\er{dtwf},~\er{hnpeq}. Each of these MLRs has a parameter~$\la\in\fik$.
\end{itemize}

\begin{remark}
\lb{rmtnib}
MTs over the NIB equation~\er{be} were studied in~\cite{GarYam2019}.
The MT~\er{bu0vb} does not appear in~\cite{GarYam2019}, 
since \er{bu0vb} is over the modified NIB equation~\er{bpeq}.

Many MTs for various equations related to~\er{be} and~\er{bpeq} can be found 
in~\cite{suris2003,BIg2016,garif2017,garif2018,GarYam2019,IgSimpl2024,MikhXenit2014,svinin2011,xenit2018} 
and references therein.
The MT~\er{bu0vb} from~\er{dtwf} to~\eqref{bpeq} 
with arbitrary constants~$\mc,\ap\in\fik$ is not equivalent 
to any of the MTs in the available literature, as far as we can see.

In the case when $\mc\neq 0$ and $\ap\neq 0$ the MT~\er{bu0vb} 
can be transformed to the MT in~\cite[page 16]{xenit2018}
by scaling and shift.
\end{remark}

Further applications of Theorem~\ref{thmuknew} 
from this~paper are presented in the joint work with E.~Chistov~\cite{ChistIg2024} 
and concern some $2$-component equations discussed below.

The $2$-component Belov--Chaltikian lattice, 
which is the Boussinesq lattice related to the lattice 
$W_3$-algebra~\cite{Belov93,hiin97,BeffaWang2013,kmw2013}, 
is the $2$-component equation
\begin{gather}
\label{bce}
\left\{
\begin{aligned}
\pd_t(\bc^1)&= \bc^1 (\bc^2_2-\bc^2_{-1}),\\ 
\pd_t(\bc^2)&= \bc^1_{-1}-\bc^1 + \bc^2 (\bc^2_1-\bc^2_{-1}),
\end{aligned}\right.
\end{gather}
where $\bc^1=\bc^1(n,t)$, $\,\bc^2=\bc^2(n,t)$.
Studying Hamiltonian evolutions of polygons,
G.~Mar\'i Beffa and Jing~Ping~Wang~\cite{BeffaWang2013} obtained 
the $2$-component equation
\begin{gather}
\label{peq}
\left\{
\begin{aligned}
\pd_t(\pe^1)&= \frac{1}{\pe^2_1}-\frac{1}{\pe^2_{-2}},\\ 
\pd_t(\pe^2)&= \frac{\pe^1_1}{\pe^2_1}-\frac{\pe^1}{\pe^2_{-1}}
\end{aligned}\right.
\end{gather}
describing the evolution induced on invariants by an invariant evolution of planar polygons. Also, it is shown in~\cite{BeffaWang2013}
that one has the following MT from~\er{peq} to~\er{bce}
\begin{gather}
\notag
\bc^1=
\frac{1}{\pe^2 \pe^2_1 \pe^2_2},\qquad\quad
\bc^2=-\frac{\pe^1_1}{\pe^2 \pe^2_1}.
\end{gather}

Using Theorem~\ref{thmuknew} from this~paper, 
in the joint work with E.~Chistov~\cite{ChistIg2024} we have constructed
a parametric family of new $2$-component integrable equations (with new MLRs)
connected by new MTs to equations~\er{bce},~\er{peq}, 
which are known to be integrable and possess MLRs.

\begin{remark}
\lb{rpmlr}
Recall that we consider MLRs for a given equation~\er{sddenew}.
We say that a MLR~$(\mm,\,\lt)$ is \emph{poor} if the matrices 
$\mm=\mm(\ud,\la)$ and $\lt=\lt(\ud,\la)$ satisfy at least one of the following conditions:
\begin{itemize}
	\item $\mm$, $\lt$ do not depend on~$\lambda$.
	\item The matrices $\mm$, $\lt$	are diagonal.
	\item $\mm$, $\lt$ do not depend on~$u_\ell$ for any $\ell\in\zz$.
\end{itemize}
Any trivial MLR is poor. 
The class of poor MLR is much wider than that of trivial ones.
For a given equation~\er{sddenew}, 
existence of a poor MLR is not sufficient for establishing integrability of~\er{sddenew},
but a poor MLR sometimes does give (limited) useful information about equation~\er{sddenew}.

For instance, let $(\mm,\,\lt)$ be a poor MLR for~\er{sddenew} such that 
\begin{itemize}
	\item $\mm$, $\lt$ are diagonal $\sm\times\sm$ matrices for some $\sm\in\zsp$,
	\item there is $j\in\{1,\ldots,\sm\}$ such that the $j$th diagonal entry~$(\mm)_{jj}$ of~$\mm$
	depends nontrivially on~$u_\ell$ for some $\ell\in\zz$.
\end{itemize}
Then, as is well known, relation~\er{lr} implies that the function $\ln\big((\mm)_{jj}\big)$
is a conserved density for equation~\er{sddenew}, since from~\er{lr} we get
\begin{gather}
\lb{lrdj}
\tdt\Big(\ln\big((\mm)_{jj}\big)\Big)=
\frac{\tdt\big((\mm)_{jj}\big)}{(\mm)_{jj}}=\cS\big((\lt)_{jj}\big)-(\lt)_{jj}=
(\cS-1)\big((\lt)_{jj}\big),
\end{gather}
where $(\lt)_{jj}$ is the $j$th diagonal entry of~$\lt$.
Furthermore, if $(\mm)_{jj}$ depends nontrivially on~$\lambda$, relation~\er{lrdj}
implies that the functions $\dfrac{d^k}{d\lambda^k}\Big(\ln\big((\mm)_{jj}\big)\Big)$, $k\in\zsp$, 
are conserved densities for~\er{sddenew} as well.
Thus such poor MLR~$(\mm,\,\lt)$ gives 
a limited amount of useful information 
about equation~\er{sddenew} in the form of conserved densities.

Theorem~\ref{thtrmlp} says that 
a MLR $(\mm,\,\lt)$ with $\cS$-part of the form~\er{muu2}
is gauge equivalent to a trivial MLR if and only if $\mm$ satisfies~\er{ctrmlp}.
It would be interesting to obtain some results on the following question: 

\noindent
Given a MLR, how to determine whether it is gauge equivalent to a poor MLR or not?
We leave this question for future work.

As said above, existence of a poor MLR is not sufficient for establishing integrability of equation~\er{sddenew}.
The examples of MLRs considered in Sections~\ref{sinb},~\ref{szhch} are not poor.

\end{remark}

\section{Results on matrix Lax representations}
\lb{srmlp}

As explained in Section~\ref{secint}, 
we have the dynamical variables 
$u_{\ell}=(u^1_{\ell},\dots,u^\diunew_{\ell})$ for $\ell\in\zz$. 

\begin{remark}
\lb{rchv}
In Theorems~\ref{thmuknew},~\ref{thtrmlp} below, 
we consider an invertible $\sm\times\sm$ matrix-function $\mm=\mm(u_0,u_1,u_2,\la)$,
and we need to choose a constant vector $a_0\in\fik^\diunew$ so that 
the matrix $\mm(a_0,u_1,u_2,\la)$ is well defined and is invertible.
Here $\mm(a_0,u_1,u_2,\la)$ is obtained from~$\mm(u_0,u_1,u_2,\la)$ 
by substituting $u_0=a_0$.

For a given invertible matrix $\mm(u_0,u_1,u_2,\la)$, such a vector $a_0\in\fik^\diunew$
can be chosen as follows. 
For $i,j\in\{1,\ldots,\sm\}$ the $(i,j)$ entry of the matrix $\mm$ 
is denoted by $\mme_{ij}=\mme_{ij}(u_0,u_1,u_2,\la)$.
Since, according to Remark~\ref{rratf}, $\mme_{ij}(u_0,u_1,u_2,\la)$ 
is assumed to be a rational function, we have
\begin{gather*}
\mme_{ij}(u_0,u_1,u_2,\la)=\frac{f_{ij}(u_0,u_1,u_2,\la)}{g_{ij}(u_0,u_1,u_2,\la)},
\qquad i,j=1,\ldots,\sm,\qquad
\det\big(\mm(u_0,u_1,u_2,\la)\big)=\frac{P(u_0,u_1,u_2,\la)}{Q(u_0,u_1,u_2,\la)},
\end{gather*}
where 
\begin{itemize}
	\item $f_{ij}(u_0,u_1,u_2,\la)$, $\,g_{ij}(u_0,u_1,u_2,\la)$, $\,P(u_0,u_1,u_2,\la)$, 
	$\,Q(u_0,u_1,u_2,\la)$ are some polynomials,
	\item each of the polynomials
\begin{gather*}
g_{ij}(u_0,u_1,u_2,\la),\qquad i,j=1,\ldots,\sm,\qquad
P(u_0,u_1,u_2,\la),\qquad Q(u_0,u_1,u_2,\la)
\end{gather*}
is not identically zero.
\end{itemize}
Then there is an open dense subset $W\subset\fik^\diunew$ such that 
for any constant vector $a_0\in W$ each of the polynomials
\begin{gather*}
g_{ij}(a_0,u_1,u_2,\la),\qquad i,j=1,\ldots,\sm,\qquad
P(a_0,u_1,u_2,\la),\qquad Q(a_0,u_1,u_2,\la)
\end{gather*}
is not identically zero. This implies that for any $a_0\in W$
the matrix $\mm(a_0,u_1,u_2,\la)$ is well defined and is invertible 
(in the sense of Remark~\ref{rratf}).

Similar considerations are applicable also in the following places:
\begin{itemize}
\item In Theorem~\ref{thmuknew}, where, for some invertible matrix
$\tilde{\mm}(u_0,u_{1},\la)$, we need to choose a constant vector 
$\tilde{a}_0\in\fik^\diunew$ so that 
the matrix $\tilde{\mm}(\tilde{a}_0,u_1,\la)$ is well defined and is invertible.
\item In Theorem~\ref{thmuk1}, 
where one considers an invertible matrix-function $\mm=\mm(u_0,u_1,\la)$.
\end{itemize}
\end{remark}

\begin{theorem}
\label{thmuknew}
Let $\sm\in\zsp$. 
Consider an invertible $\sm\times\sm$ matrix-function $\mm=\mm(u_0,u_1,u_2,\la)$. 
Suppose that 
\begin{gather}
\lb{pdmnew}
\forall\, i,j=1,\dots,\diunew\qquad\quad
\dd{u^i_0}\Big(\dd{u^j_2}
\big(\mm(u_0,u_1,u_2,\la)\big)\cdot \mm(u_0,u_1,u_2,\la)^{-1}\Big)=0,\\
\lb{pdm01}
\begin{split}
\forall\, i,j&=1,\dots,\diunew\qquad
\dd{u^i_0}
\left(\dd{u^j_1}\big(\mm(u_0,u_1,u_2,\la)\big)\cdot\mm(u_0,u_1,u_2,\la)^{-1}+\right.\\
&\left.+\mm(u_0,u_1,u_2,\la)\cdot
\dd{u^j_1}\big(\mm(a_0,u_0,u_1,\la)\big)\cdot\mm(a_0,u_0,u_1,\la)^{-1}
\cdot\mm(u_0,u_1,u_2,\la)^{-1}\right)=0,
\end{split}
\end{gather}
where $a_0\in\fik^\diunew$ is a constant vector and\/ 
$\mm(a_0,u_0,u_1,\la)=\cS^{-1}\big(\mm(a_0,u_1,u_2,\la)\big)$.
The matrix\/ 
$\mm(a_0,u_1,u_2,\la)$ is obtained from~$\mm(u_0,u_1,u_2,\la)$ 
by substituting $u_0=a_0$.
Here a vector $a_0\in\fik^\diunew$ is chosen so that the matrix\/ 
$\mm(a_0,u_1,u_2,\la)$ is well defined and is invertible.
Remark~\textup{\ref{rchv}} explains how to choose such a vector.

Condition~\er{pdmnew} implies that the matrix-function 
\begin{gather}
\lb{tilm}
\tilde{\mm}=\mm(a_0,u_1,u_2,\la)^{-1}\cdot\mm(u_0,u_1,u_2,\la)\cdot\mm(a_0,u_0,u_1,\la)
\end{gather}
does not depend on~$u_2$.
Thus $\tilde{\mm}$ is of the form $\tilde{\mm}=\tilde{\mm}(u_0,u_{1},\la)$. 
Consider the gauge transformation
\begin{gather}
\lb{guu1}
\mathbf{g}(u_0,u_{1},\la)=
\tilde{\mm}(\tilde{a}_0,u_0,\la)^{-1}\cdot\mm(a_0,u_0,u_1,\la)^{-1},
\end{gather}
where $\tilde{a}_0\in\fik^\diunew$ is another constant vector and\/ 
$\tilde{\mm}(\tilde{a}_0,u_0,\la)=\cS^{-1}\big(\tilde{\mm}(\tilde{a}_0,u_1,\la)\big)$.
Here $\tilde{a}_0\in\fik^\diunew$ is chosen so that $\tilde{\mm}(\tilde{a}_0,u_1,\la)$ is well defined and is invertible.

Conditions~\er{pdmnew},~\er{pdm01} imply that the matrix-function 
\begin{gather}
\lb{hatmm}
\hat{\mm}=
\cS\big(\mathbf{g}(u_0,u_{1},\la)\big)\cdot\mm(u_0,u_1,u_2,\la)\cdot\mathbf{g}(u_0,u_{1},\la)^{-1}
\end{gather}
does not depend on~$u_1,\,u_2$.
Thus $\hat{\mm}$ is of the form $\hat{\mm}=\hat{\mm}(u_0,\la)$. 

Now suppose that we have a MLR $(\mm,\,\lt)$ with~$\mm=\mm(u_0,u_1,u_2,\la)$.
Conditions~\er{pdmnew},~\er{pdm01} are sufficient for the possibility to transform 
\textup{(}by means of a gauge transformation\textup{)} the MLR $(\mm,\,\lt)$ to a MLR 
$(\hat{\mm},\,\hat{\lt})$ with $\cS$-part of the form $\hat{\mm}=\hat{\mm}(u_0,\la)$.
One can compute $(\hat{\mm},\,\hat{\lt})$ by means of~\er{hatmm} and 
\begin{gather}
\lb{hatlt}
\hat{\lt}=\tdt\big(\mathbf{g}(u_0,u_{1},\la)\big)\cdot\mathbf{g}(u_0,u_{1},\la)^{-1}+
\mathbf{g}(u_0,u_{1},\la)\cdot\lt\cdot\mathbf{g}(u_0,u_{1},\la)^{-1},
\end{gather}
where the gauge transformation $\mathbf{g}(u_0,u_{1},\la)$ is given by~\er{guu1}.
\end{theorem}
\begin{proof}
Let $j\in\{1,\dots,\diunew\}$. 
Using~\er{tilm}, we obtain
\begin{multline}
\lb{pdtilmn}
\dd{u^j_2}\big(\tilde{\mm}\big)
=\dd{u^j_2}\big(\mm(a_0,u_1,u_2,\la)^{-1}\cdot\mm(u_0,u_1,u_2,\la)\big)
\cdot\mm(a_0,u_0,u_1,\la)=\\
=\Big(\dd{u^j_2}\big(\mm(a_0,u_1,u_2,\la)^{-1}\big)\cdot \mm(u_0,u_1,u_2,\la)
+\mm(a_0,u_1,u_2,\la)^{-1}\cdot\dd{u^j_2}\big(\mm(u_0,u_1,u_2,\la)\big)\Big)
\cdot\mm(a_0,u_0,u_1,\la)=\\
=\Big(-\mm(a_0,u_1,u_2,\la)^{-1}\cdot\dd{u^j_2}\big(\mm(a_0,u_1,u_2,\la)\big)\cdot 
\mm(a_0,u_1,u_2,\la)^{-1}\cdot \mm(u_0,u_1,u_2,\la)+\\
+\mm(a_0,u_1,u_2,\la)^{-1}\cdot\dd{u^j_2}\big(\mm(u_0,u_1,u_2,\la)\big)\Big)
\cdot\mm(a_0,u_0,u_1,\la)=\\
=\mm(a_0,u_1,u_2,\la)^{-1}\cdot\Big(-\dd{u^j_2}\big(\mm(a_0,u_1,u_2,\la)\big)\cdot
\mm(a_0,u_1,u_2,\la)^{-1}+\\
+\dd{u^j_2}\big(\mm(u_0,u_1,u_2,\la)\big)\cdot
\mm(u_0,u_1,u_2,\la)^{-1}\Big)\cdot \mm(u_0,u_1,u_2,\la)\cdot\mm(a_0,u_0,u_1,\la)=\\
=\mm(a_0,u_1,u_2,\la)^{-1}\cdot Y^j(u_0,u_1,u_2,\la)\cdot \mm(u_0,u_1,u_2,\la)\cdot\mm(a_0,u_0,u_1,\la),
\end{multline}
where 
\begin{gather}
\lb{Lunew}
Y^j(u_0,u_1,u_2,\la)=-\dd{u^j_2}\big(\mm(a_0,u_1,u_2,\la)\big)\cdot
\mm(a_0,u_1,u_2,\la)^{-1}+\dd{u^j_2}\big(\mm(u_0,u_1,u_2,\la)\big)\cdot \mm(u_0,u_1,u_2,\la)^{-1}.
\end{gather}
From~\er{Lunew} and~\er{pdmnew} it follows that
\begin{gather}
\lb{pdLnew}
\dd{u^i_0}\big(Y^j(u_0,u_1,u_2,\la)\big)=
\dd{u^i_0}\Big(\dd{u^j_2}
\big(\mm(u_0,u_1,u_2,\la)\big)\cdot \mm(u_0,u_1,u_2,\la)^{-1}\Big)=0\qquad
\forall\, i=1,\dots,\diunew,\\
\lb{La0new}
Y^j(a_0,u_1,u_2,\la)=0.
\end{gather}

Since, according to Remark~\ref{rratf}, the matrix-function $Y^j(u_0,u_1,u_2,\la)$ is rational,
equations~\er{pdLnew},~\er{La0new} imply that $Y^j(u_0,u_1,u_2,\la)$ is identically zero.
(From~\er{pdLnew},~\er{La0new} it follows that the matrix-function $Y^j(u_0,u_1,u_2,\la)$ 
is zero on some non-empty open subset of the space with coordinates~$u_0,u_1,u_2,\la$.
Since $Y^j(u_0,u_1,u_2,\la)$ is rational, this implies that $Y^j(u_0,u_1,u_2,\la)$ is identically zero.)

Substituting $Y^j(u_0,u_1,u_2,\la)=0$ in~\er{pdtilmn}, one obtains
$\dfrac{\pd}{\pd u^j_2}\big(\tilde{\mm}\big)=0$ for all $j=1,\dots,\diunew$.
Hence the matrix-function \er{tilm} is of the form $\tilde{\mm}=\tilde{\mm}(u_0,u_{1},\la)$. 
Using~\er{tilm}, for each $j\in\{1,\dots,\diunew\}$ one gets
\begin{multline}
\lb{tmlong}
\dd{u^j_1}
\big(\tilde{\mm}(u_0,u_{1},\la)\big)=
\dd{u^j_1}
\Big(\mm(a_0,u_1,u_2,\la)^{-1}\cdot\mm(u_0,u_1,u_2,\la)\cdot\mm(a_0,u_0,u_1,\la)\Big)=\\
=-\mm(a_0,u_1,u_2,\la)^{-1}\cdot\dd{u^j_1}\big(\mm(a_0,u_1,u_2,\la)\big)\cdot\mm(a_0,u_1,u_2,\la)^{-1}\cdot
\mm(u_0,u_1,u_2,\la)\cdot\mm(a_0,u_0,u_1,\la)+\\
+\mm(a_0,u_1,u_2,\la)^{-1}\cdot\dd{u^j_1}
\big(\mm(u_0,u_1,u_2,\la)\big)\cdot\mm(a_0,u_0,u_1,\la)+\\
+\mm(a_0,u_1,u_2,\la)^{-1}\cdot\mm(u_0,u_1,u_2,\la)\cdot\dd{u^j_1}\big(\mm(a_0,u_0,u_1,\la)\big).
\end{multline}
From~\er{tilm},~\er{tmlong} we derive
\begin{multline}
\lb{dtmtm}
\dd{u^j_1}
\big(\tilde{\mm}(u_0,u_{1},\la)\big)\cdot\tilde{\mm}(u_0,u_{1},\la)^{-1}=\\
=
\dd{u^j_1}
\big(\tilde{\mm}(u_0,u_{1},\la)\big)\cdot\mm(a_0,u_0,u_1,\la)^{-1}\cdot\mm(u_0,u_1,u_2,\la)^{-1}\cdot\mm(a_0,u_1,u_2,\la)=\\
=-\mm(a_0,u_1,u_2,\la)^{-1}\cdot\dd{u^j_1}\big(\mm(a_0,u_1,u_2,\la)\big)\cdot\mm(a_0,u_1,u_2,\la)^{-1}+\\
+\mm(a_0,u_1,u_2,\la)^{-1}\cdot\left(\dd{u^j_1}
\big(\mm(u_0,u_1,u_2,\la)\big)\cdot\mm(u_0,u_1,u_2,\la)^{-1}+\right.\\
\left.+\mm(u_0,u_1,u_2,\la)\cdot\dd{u^j_1}\big(\mm(a_0,u_0,u_1,\la)\big)\cdot
\mm(a_0,u_0,u_1,\la)^{-1}\cdot\mm(u_0,u_1,u_2,\la)^{-1}\right)\cdot\mm(a_0,u_1,u_2,\la).
\end{multline}
From~\er{dtmtm},~\er{pdm01} it follows that
\begin{multline}
\lb{pdtm}
\forall\, i,j=1,\dots,\diunew\qquad
\dd{u^i_0}\Big(\dd{u^j_1}
\big(\tilde{\mm}(u_0,u_{1},\la)\big)\cdot\tilde{\mm}(u_0,u_{1},\la)^{-1}\Big)=\\
=\mm(a_0,u_1,u_2,\la)^{-1}\cdot\dd{u^i_0}\left(\dd{u^j_1}
\big(\mm(u_0,u_1,u_2,\la)\big)\cdot\mm(u_0,u_1,u_2,\la)^{-1}+\right.\\
\left.+\mm(u_0,u_1,u_2,\la)\cdot\dd{u^j_1}\big(\mm(a_0,u_0,u_1,\la)\big)\cdot
\mm(a_0,u_0,u_1,\la)^{-1}\cdot\mm(u_0,u_1,u_2,\la)^{-1}\right)\cdot\mm(a_0,u_1,u_2,\la)=0.
\end{multline}

Using~\er{tilm} and~\er{guu1}, we see that the matrix-function~\er{hatmm} can be written as
\begin{gather}
\lb{hatmtm}
\hat{\mm}=\tilde{\mm}(\tilde{a}_0,u_1,\la)^{-1}\cdot\tilde{\mm}(u_0,u_{1},\la)\cdot
\tilde{\mm}(\tilde{a}_0,u_0,\la).
\end{gather}
Formula~\er{hatmtm} implies that $\hat{\mm}$ does not depend on~$u_2$.

It remains to show that $\hat{\mm}$ does not depend on~$u_1$.
Let $j\in\{1,\dots,\diunew\}$. 
Using~\er{hatmtm}, we get
\begin{multline}
\lb{pdhatm}
\dd{u^j_1}\big(\hat{\mm}\big)
=\dd{u^j_1}\big(\tilde{\mm}(\tilde{a}_0,u_1,\la)^{-1}\cdot\tilde{\mm}(u_0,u_{1},\la)\big)
\cdot\tilde{\mm}(\tilde{a}_0,u_0,\la)=\\
=\Big(\dd{u^j_1}\big(\tilde{\mm}(\tilde{a}_0,u_1,\la)^{-1}\big)\cdot \tilde{\mm}(u_0,u_1,\la)
+\tilde{\mm}(\tilde{a}_0,u_1,\la)^{-1}\cdot\dd{u^j_1}\big(\tilde{\mm}(u_0,u_1,\la)\big)\Big)
\cdot\tilde{\mm}(\tilde{a}_0,u_0,\la)=\\
=\Big(-\tilde{\mm}(\tilde{a}_0,u_1,\la)^{-1}\cdot\dd{u^j_1}\big(\tilde{\mm}(\tilde{a}_0,u_1,\la)\big)\cdot 
\tilde{\mm}(\tilde{a}_0,u_1,\la)^{-1}\cdot \tilde{\mm}(u_0,u_1,\la)+\\
+\tilde{\mm}(\tilde{a}_0,u_1,\la)^{-1}\cdot\dd{u^j_1}\big(\tilde{\mm}(u_0,u_1,\la)\big)\Big)
\cdot\tilde{\mm}(\tilde{a}_0,u_0,\la)=\\
=\tilde{\mm}(\tilde{a}_0,u_1,\la)^{-1}\cdot\Big(-\dd{u^j_1}\big(\tilde{\mm}(\tilde{a}_0,u_1,\la)\big)\cdot
\tilde{\mm}(\tilde{a}_0,u_1,\la)^{-1}+\\
+\dd{u^j_1}\big(\tilde{\mm}(u_0,u_1,\la)\big)\cdot
\tilde{\mm}(u_0,u_1,\la)^{-1}\Big)\cdot \tilde{\mm}(u_0,u_1,\la)\cdot\tilde{\mm}(\tilde{a}_0,u_0,\la)=\\
=\tilde{\mm}(\tilde{a}_0,u_1,\la)^{-1}\cdot\tilde{Y}^j(u_0,u_1,\la)\cdot 
\tilde{\mm}(u_0,u_1,\la)\cdot\tilde{\mm}(\tilde{a}_0,u_0,\la),
\end{multline}
where 
\begin{gather}
\lb{tLu}
\tilde{Y}^j(u_0,u_1,\la)=-\dd{u^j_1}\big(\tilde{\mm}(\tilde{a}_0,u_1,\la)\big)\cdot
\tilde{\mm}(\tilde{a}_0,u_1,\la)^{-1}+\dd{u^j_1}\big(\tilde{\mm}(u_0,u_1,\la)\big)
\cdot\tilde{\mm}(u_0,u_1,\la)^{-1}.
\end{gather}
From~\er{tLu},~\er{pdtm} it follows that
\begin{gather}
\lb{pdtL}
\dd{u^i_0}\big(\tilde{Y}^j(u_0,u_1,\la)\big)=
\dd{u^i_0}\Big(\dd{u^j_1}
\big(\tilde{\mm}(u_0,u_1,\la)\big)\cdot\tilde{\mm}(u_0,u_1,\la)^{-1}\Big)=0\qquad
\forall\, i=1,\dots,\diunew,\\
\lb{tLa0}
\tilde{Y}^j(\tilde{a}_0,u_1,\la)=0.
\end{gather}

According to Remark~\ref{rratf}, the matrix-function $\tilde{Y}^j(u_0,u_1,\la)$ is rational.
Then equations~\er{pdtL},~\er{tLa0} imply that $\tilde{Y}^j(u_0,u_1,\la)$ is identically zero.
(This is shown similarly to the above explanation of the fact that 
equations~\er{pdLnew},~\er{La0new} imply the identity $Y^j(u_0,u_1,u_2,\la)=0$.)

Substituting $\tilde{Y}^j(u_0,u_1,\la)=0$ in~\er{pdhatm}, one obtains
$\dfrac{\pd}{\pd u^j_1}\big(\hat{\mm}\big)=0$ for all $j=1,\dots,\diunew$.
Thus $\hat{\mm}$ does not depend on~$u_1$.

Now suppose that we have a MLR $(\mm,\,\lt)$ 
with~$\mm=\mm(u_0,u_1,u_2,\la)$ satisfying~\er{pdmnew},~\er{pdm01}.
Applying the gauge transformation~\er{guu1} to this MLR, we get the MLR
\begin{gather*}
\hat{\mm}=\cS\big(\mathbf{g}(u_0,u_{1},\la)\big)\cdot\mm(u_0,u_1,u_2,\la)\cdot\mathbf{g}(u_0,u_{1},\la)^{-1},\\
\hat{\lt}=\tdt\big(\mathbf{g}(u_0,u_{1},\la)\big)\cdot\mathbf{g}(u_0,u_{1},\la)^{-1}+
\mathbf{g}(u_0,u_{1},\la)\cdot\lt\cdot\mathbf{g}(u_0,u_{1},\la)^{-1}.
\end{gather*}
As shown above, $\hat{\mm}$ is of the form $\hat{\mm}=\hat{\mm}(u_0,\la)$.
\end{proof}

\begin{theorem}
\label{thtrmlp}
Consider a MLR\/ 
$(\mm,\,\lt)$ with $\cS$-part of the form $\mm=\mm(u_0,u_1,u_2,\la)$.
This MLR is gauge equivalent to a trivial MLR if and only if\/ 
$\mm$ satisfies
\begin{gather}
\lb{ctrmlp}
\forall\,j=1,\dots,\diunew\qquad 
\dd{u^j_2}(\mm)\cdot \mm^{-1}=
-\cS\Big(\mm^{-1}\cdot\Big(\dd{u^j_1}(\mm)\cdot\mm^{-1}+
\cS\Big(\mm^{-1}\cdot\dd{u^j_0}(\mm)\Big)\Big)\cdot \mm\Big).
\end{gather}
\end{theorem}
\begin{proof}
Suppose that a MLR $(\mm,\,\lt)$ with $\cS$-part $\mm=\mm(u_0,u_1,u_2,\la)$ 
is gauge equivalent to a trivial MLR. This means that there is a 
gauge transformation $\mathbf{g}=\mathbf{g}(\ud,\la)$ such that the matrix
\begin{gather}
\lb{hatmp}
\hat{\mm}=\cS(\mathbf{g})\cdot\mm(u_0,u_1,u_2,\la)\cdot\mathbf{g}^{-1}
\end{gather}
obeys 
\begin{gather}
\lb{dujhm}
\dd{u^j_\ell}(\hat{\mm})=0\qquad\quad\forall\,j,\ell.
\end{gather}
Relations~\er{hatmp},~\er{dujhm} imply that the matrix $\mathbf{g}$ 
may depend only on $u_0,u_1,\la$. That is, 
\begin{gather}
\lb{gguu}
\mathbf{g}=\mathbf{g}(u_0,u_1,\la).
\end{gather}

From~\er{hatmp} we get 
\begin{gather}
\lb{mmghm}
\mm=\cS(\mathbf{g}^{-1})\cdot\hat{\mm}\cdot\mathbf{g}.
\end{gather}
Using~\er{dujhm},~\er{gguu},~\er{mmghm}, for each $j=1,\dots,\diunew$ one obtains
\begin{gather}
\lb{duj2}
\dd{u^j_2}(\mm)=\cS\big(\dd{u^j_1}(\mathbf{g}^{-1})\big)\cdot\hat{\mm}\cdot\mathbf{g}=
-\cS(\mathbf{g}^{-1})\cdot\cS\big(\dd{u^j_1}(\mathbf{g})\big)
\cdot\cS(\mathbf{g}^{-1})\cdot\hat{\mm}\cdot\mathbf{g},\\
\lb{duj1}
\dd{u^j_1}(\mm)=-\cS(\mathbf{g}^{-1})\cdot\cS\big(\dd{u^j_0}(\mathbf{g})\big)
\cdot\cS(\mathbf{g}^{-1})\cdot\hat{\mm}\cdot\mathbf{g}+
\cS(\mathbf{g}^{-1})\cdot\hat{\mm}\cdot\dd{u^j_1}(\mathbf{g}),\\
\lb{duj0}
\dd{u^j_0}(\mm)=\cS(\mathbf{g}^{-1})\cdot\hat{\mm}\cdot\dd{u^j_0}(\mathbf{g}).
\end{gather}
From \er{mmghm}--\er{duj0} it follows that
\begin{gather}
\lb{uj2m}
\dd{u^j_2}(\mm)\cdot \mm^{-1}=-\cS(\mathbf{g}^{-1})\cdot\cS\big(\dd{u^j_1}(\mathbf{g})\big)=\cS\big(-\mathbf{g}^{-1}\cdot\dd{u^j_1}(\mathbf{g})\big),\\
\lb{muj0m}
\mm^{-1}\cdot\dd{u^j_0}(\mm)=\mathbf{g}^{-1}\cdot\dd{u^j_0}(\mathbf{g}),\\
\lb{muj1m}
\mm^{-1}\cdot\dd{u^j_1}(\mm)=-\mathbf{g}^{-1}\cdot\hat{\mm}^{-1}\cdot\cS\big(\dd{u^j_0}(\mathbf{g})\big)
\cdot\cS(\mathbf{g}^{-1})\cdot\hat{\mm}\cdot\mathbf{g}+
\mathbf{g}^{-1}\cdot\dd{u^j_1}(\mathbf{g}).
\end{gather}
Using~\er{mmghm}, \er{uj2m}, \er{muj0m}, \er{muj1m}, one gets
\begin{multline}
\lb{mmbdd}
\mm^{-1}\cdot\Big(\dd{u^j_1}(\mm)\cdot\mm^{-1}+
\cS\Big(\mm^{-1}\cdot\dd{u^j_0}(\mm)\Big)\Big)\cdot \mm=\\
=\mm^{-1}\cdot\dd{u^j_1}(\mm)+
\mm^{-1}\cdot\cS\Big(\mm^{-1}\cdot\dd{u^j_0}(\mm)\Big)\cdot \mm=\\
=-\mathbf{g}^{-1}\cdot\hat{\mm}^{-1}\cdot\cS\big(\dd{u^j_0}(\mathbf{g})\big)
\cdot\cS(\mathbf{g}^{-1})\cdot\hat{\mm}\cdot\mathbf{g}
+\mathbf{g}^{-1}\cdot\dd{u^j_1}(\mathbf{g})
+\mathbf{g}^{-1}\cdot\hat{\mm}^{-1}\cdot\cS(\mathbf{g})\cdot
\cS\Big(\mathbf{g}^{-1}\cdot\dd{u^j_0}(\mathbf{g})\Big)\cdot
\cS(\mathbf{g}^{-1})\cdot\hat{\mm}\cdot\mathbf{g}.
\end{multline}
Since 
\begin{multline*}
\mathbf{g}^{-1}\cdot\hat{\mm}^{-1}\cdot\cS(\mathbf{g})\cdot
\cS\Big(\mathbf{g}^{-1}\cdot\dd{u^j_0}(\mathbf{g})\Big)\cdot
\cS(\mathbf{g}^{-1})\cdot\hat{\mm}\cdot\mathbf{g}=\\
=\mathbf{g}^{-1}\cdot\hat{\mm}^{-1}\cdot\cS(\mathbf{g})\cdot
\cS(\mathbf{g}^{-1})\cdot\cS\big(\dd{u^j_0}(\mathbf{g})\big)
\cdot\cS(\mathbf{g}^{-1})\cdot\hat{\mm}\cdot\mathbf{g}=
\mathbf{g}^{-1}\cdot\hat{\mm}^{-1}\cdot\cS\big(\dd{u^j_0}(\mathbf{g})\big)
\cdot\cS(\mathbf{g}^{-1})\cdot\hat{\mm}\cdot\mathbf{g},
\end{multline*}
from~\er{mmbdd} we derive
\begin{gather}
\lb{mmdds}
\mm^{-1}\cdot\Big(\dd{u^j_1}(\mm)\cdot\mm^{-1}+
\cS\Big(\mm^{-1}\cdot\dd{u^j_0}(\mm)\Big)\Big)\cdot \mm=\mathbf{g}^{-1}\cdot\dd{u^j_1}(\mathbf{g}).
\end{gather}
Then equations~\er{uj2m},~\er{mmdds} imply that $\mm$ satisfies~\er{ctrmlp}.

Now, conversely, consider 
a MLR $(\mm,\,\lt)$ with $\cS$-part $\mm=\mm(u_0,u_1,u_2,\la)$ 
satisfying~\er{ctrmlp}. We need to prove that 
this MLR is gauge equivalent to a trivial MLR. 
To construct a suitable gauge transformation, we are going to 
show that $\mm$ obeys the conditions of Theorem~\ref{thmuknew}.

Since $\mm$ does not depend on~$u_\ell$ for $\ell<0$, 
the right-hand side of~\er{ctrmlp} does not depend on~$u_0$.
Then \er{ctrmlp} implies that $\mm$ satisfies~\er{pdmnew}.
Using the rules~\er{csuf}, we see that equation~\er{ctrmlp} reads
\begin{multline}
\lb{ctrmlpe}
\dd{u^j_2}\big(\mm(u_0,u_1,u_2,\la)\big)\cdot \mm(u_0,u_1,u_2,\la)^{-1}=\\
=-\mm(u_1,u_2,u_3,\la)^{-1}\cdot\Big(\dd{u^j_2}\big(\mm(u_1,u_2,u_3,\la)\big)
\cdot\mm(u_1,u_2,u_3,\la)^{-1}+\\
+\mm(u_2,u_3,u_4,\la)^{-1}\cdot\dd{u^j_2}\big(\mm(u_2,u_3,u_4,\la)\big)\Big)\cdot \mm(u_1,u_2,u_3,\la).
\end{multline}
Now we take a constant vector $a_0\in\fik^\diunew$, substitute $u_0=a_0$ in~\er{ctrmlpe}, 
and after that we apply the operator~$\cS^{-1}$ to equation~\er{ctrmlpe}. 
This gives the following
\begin{multline}
\lb{ctrmlps}
\dd{u^j_1}\big(\mm(a_0,u_0,u_1,\la)\big)\cdot \mm(a_0,u_0,u_1,\la)^{-1}=\\
=-\mm(u_0,u_1,u_2,\la)^{-1}\cdot\Big(\dd{u^j_1}\big(\mm(u_0,u_1,u_2,\la)\big)
\cdot\mm(u_0,u_1,u_2,\la)^{-1}+\\
+\mm(u_1,u_2,u_3,\la)^{-1}\cdot\dd{u^j_1}\big(\mm(u_1,u_2,u_3,\la)\big)\Big)\cdot \mm(u_0,u_1,u_2,\la),
\end{multline}
where $\mm(a_0,u_0,u_1,\la)=\cS^{-1}\big(\mm(a_0,u_1,u_2,\la)\big)$.
Now we multiply~\er{ctrmlps} by $\mm(u_0,u_1,u_2,\la)$ from the left 
and by $\mm(u_0,u_1,u_2,\la)^{-1}$ from the right.
This gives the equation 
\begin{multline}
\lb{ctrmlpm}
\mm(u_0,u_1,u_2,\la)\cdot
\dd{u^j_1}\big(\mm(a_0,u_0,u_1,\la)\big)\cdot \mm(a_0,u_0,u_1,\la)^{-1}
\cdot\mm(u_0,u_1,u_2,\la)^{-1}=\\
=-\dd{u^j_1}\big(\mm(u_0,u_1,u_2,\la)\big)\cdot\mm(u_0,u_1,u_2,\la)^{-1}
+\mm(u_1,u_2,u_3,\la)^{-1}\cdot\dd{u^j_1}\big(\mm(u_1,u_2,u_3,\la)\big).
\end{multline}
Adding $\dd{u^j_1}\big(\mm(u_0,u_1,u_2,\la)\big)\cdot\mm(u_0,u_1,u_2,\la)^{-1}$ 
to the both sides of equation~\er{ctrmlpm}, one obtains
\begin{multline}
\lb{ctrmlpa}
\dd{u^j_1}\big(\mm(u_0,u_1,u_2,\la)\big)\cdot\mm(u_0,u_1,u_2,\la)^{-1}+\\
+\mm(u_0,u_1,u_2,\la)\cdot
\dd{u^j_1}\big(\mm(a_0,u_0,u_1,\la)\big)\cdot \mm(a_0,u_0,u_1,\la)^{-1}
\cdot\mm(u_0,u_1,u_2,\la)^{-1}=\\
=\mm(u_1,u_2,u_3,\la)^{-1}\cdot\dd{u^j_1}\big(\mm(u_1,u_2,u_3,\la)\big).
\end{multline}
Equation~\er{ctrmlpa} implies that $\mm$ obeys~\er{pdm01}.

Then, by Theorem~\ref{thmuknew}, 
there is a gauge transformation $\mathbf{g}(u_0,u_{1},\la)$ such that the matrix-function 
\begin{gather}
\lb{chmm}
\check{\mm}=
\cS\big(\mathbf{g}(u_0,u_{1},\la)\big)\cdot\mm(u_0,u_1,u_2,\la)\cdot\mathbf{g}(u_0,u_{1},\la)^{-1}
\end{gather}
does not depend on~$u_1,\,u_2$. 
Thus $\check{\mm}$ is of the form $\check{\mm}=\check{\mm}(u_0,\la)$.
From~\er{chmm} we obtain 
\begin{gather}
\lb{mgchm}
\mm=\cS\big(\mathbf{g}(u_0,u_{1},\la)^{-1}\big)\cdot
\check{\mm}(u_0,\la)\cdot\mathbf{g}(u_0,u_{1},\la).
\end{gather}
Using~\er{mgchm}, for each $j=1,\dots,\diunew$ one gets
\begin{gather}
\lb{nuj2mch}
\dd{u^j_2}(\mm)\cdot \mm^{-1}=-\cS(\mathbf{g}^{-1})\cdot\cS\big(\dd{u^j_1}(\mathbf{g})\big)=\cS\big(-\mathbf{g}^{-1}\cdot\dd{u^j_1}(\mathbf{g})\big),\\
\lb{nmuj1m}
\mm^{-1}\cdot\dd{u^j_1}(\mm)=
-\mathbf{g}^{-1}\cdot\check{\mm}^{-1}\cdot\cS\big(\dd{u^j_0}(\mathbf{g})\big)
\cdot\cS(\mathbf{g}^{-1})\cdot\check{\mm}\cdot\mathbf{g}+
\mathbf{g}^{-1}\cdot\dd{u^j_1}(\mathbf{g}),\\
\lb{nmuj0m}
\mm^{-1}\cdot\dd{u^j_0}(\mm)=
\mathbf{g}^{-1}\cdot\check{\mm}^{-1}\cdot
\dd{u^j_0}(\check{\mm})\cdot\mathbf{g}
+\mathbf{g}^{-1}\cdot\dd{u^j_0}(\mathbf{g}).
\end{gather}
Applying~ the operator $-\cS^{-1}$ to equation~\er{ctrmlp} 
and using~\er{mgchm}--\er{nmuj0m}, we derive
\begin{multline}
\lb{longm}
\mathbf{g}^{-1}\cdot\dd{u^j_1}(\mathbf{g})=
\mm^{-1}\cdot\Big(\dd{u^j_1}(\mm)\cdot\mm^{-1}+
\cS\Big(\mm^{-1}\cdot\dd{u^j_0}(\mm)\Big)\Big)\cdot \mm=\\
=\mm^{-1}\cdot\dd{u^j_1}(\mm)+
\mm^{-1}\cdot\cS\Big(\mathbf{g}^{-1}\cdot\check{\mm}^{-1}\cdot
\dd{u^j_0}(\check{\mm})\cdot\mathbf{g}
+\mathbf{g}^{-1}\cdot\dd{u^j_0}(\mathbf{g})\Big)\cdot \mm=\\
=-\mathbf{g}^{-1}\cdot\check{\mm}^{-1}\cdot\cS\big(\dd{u^j_0}(\mathbf{g})\big)
\cdot\cS(\mathbf{g}^{-1})\cdot\check{\mm}\cdot\mathbf{g}+
\mathbf{g}^{-1}\cdot\dd{u^j_1}(\mathbf{g})+\\
+\mm^{-1}\cdot\cS\Big(\mathbf{g}^{-1}\cdot\check{\mm}^{-1}\cdot
\dd{u^j_0}(\check{\mm})\cdot\mathbf{g}\Big)\cdot \mm+
\mm^{-1}\cdot\cS\big(\mathbf{g}^{-1}\big)\cdot\cS\big(\dd{u^j_0}(\mathbf{g})\big)\cdot\mm=\\
=\mathbf{g}^{-1}\cdot\dd{u^j_1}(\mathbf{g})
+\mm^{-1}\cdot\cS\Big(\mathbf{g}^{-1}\cdot\check{\mm}^{-1}\cdot
\dd{u^j_0}(\check{\mm})\cdot\mathbf{g}\Big)\cdot \mm.
\end{multline}
From~\er{longm} one obtains the equation 
$\mm^{-1}\cdot\cS\Big(\mathbf{g}^{-1}\cdot\check{\mm}^{-1}\cdot
\dd{u^j_0}(\check{\mm})\cdot\mathbf{g}\Big)\cdot \mm=0$, 
which implies $\dd{u^j_0}(\check{\mm})=0$.

Thus we have~\er{chmm}, where $\dd{u^j_\ell}(\check{\mm})=0$ for all $j,\,\ell$.
Hence the MLR~$(\mm,\,\lt)$ is gauge equivalent to a trivial MLR. 
\end{proof}

Theorem~\ref{thmuk1} below can be deduced from~\cite[Theorem 1]{IgSimpl2024}, 
but for completeness we present a proof for it.

\begin{theorem}
\label{thmuk1}
Let $\sm\in\zsp$. Consider an invertible $\sm\times\sm$ matrix-function~$\mm=\mm(u_0,u_1,\la)$. 
Suppose that 
\begin{gather}
\lb{pdm1}
\forall\, i,j=1,\dots,\diunew\qquad\quad
\dd{u^i_0}\Big(\dd{u^j_1}
\big(\mm(u_0,u_1,\la)\big)\cdot \mm(u_0,u_1,\la)^{-1}\Big)=0.
\end{gather}
Consider the gauge transformation
\begin{gather}
\lb{guu0}
\mathbf{g}(u_0,\la)=
\cS^{-1}\big(\mm(a_0,u_1,\la)^{-1}\big),
\end{gather}
where $a_0\in\fik^\diunew$ is a constant vector and the matrix\/ 
$\mm(a_0,u_1,\la)$ is obtained from~$\mm(u_0,u_1,\la)$ by substituting $u_0=a_0$.
Here a vector $a_0\in\fik^\diunew$ is chosen so that the matrix\/ 
$\mm(a_0,u_1,\la)$ is well defined and is invertible.
Remark~\textup{\ref{rchv}} explains how to choose such a vector.

Condition~\er{pdm1} implies that the matrix-function 
\begin{gather}
\lb{hatmm1}
\hat{\mm}=\cS\big(\mathbf{g}(u_0,\la)\big)\cdot\mm(u_0,u_1,\la)\cdot\mathbf{g}(u_0,\la)^{-1}
\end{gather}
does not depend on~$u_1$.
Thus $\hat{\mm}$ is of the form $\hat{\mm}=\hat{\mm}(u_0,\la)$. 

Now suppose that we have a MLR $(\mm,\,\lt)$ with~$\mm=\mm(u_0,u_1,\la)$.
Condition~\er{pdm1} is sufficient for the possibility to transform 
\textup{(}by means of a gauge transformation\textup{)} the MLR $(\mm,\,\lt)$ to a MLR 
$(\hat{\mm},\,\hat{\lt})$ with $\cS$-part of the form $\hat{\mm}=\hat{\mm}(u_0,\la)$.
One can compute $(\hat{\mm},\,\hat{\lt})$ by means of~\er{hatmm1} and 
\begin{gather}
\lb{hatlt1}
\hat{\lt}=\tdt\big(\mathbf{g}(u_0,\la)\big)\cdot\mathbf{g}(u_0,\la)^{-1}+
\mathbf{g}(u_0,\la)\cdot\lt\cdot\mathbf{g}(u_0,\la)^{-1},
\end{gather}
where the gauge transformation $\mathbf{g}(u_0,\la)$ is given by~\er{guu0}.
\end{theorem}
\begin{proof}
Substituting~\er{guu0} in~\er{hatmm1}, one gets
\begin{gather}
\lb{hatmms}
\hat{\mm}=\cS\big(\mathbf{g}(u_0,\la)\big)\cdot\mm(u_0,u_1,\la)\cdot\mathbf{g}(u_0,\la)^{-1}=
\mm(a_0,u_1,\la)^{-1}\cdot\mm(u_0,u_1,\la)\cdot\mm(a_0,u_0,\la),
\end{gather}
where $\mm(a_0,u_0,\la)=\cS^{-1}\big(\mm(a_0,u_1,\la)\big)$.
Let $j\in\{1,\dots,\diunew\}$. Using~\er{hatmms}, we obtain
\begin{multline}
\lb{pdtilmn1}
\dd{u^j_1}\big(\hat{\mm}\big)
=\dd{u^j_1}\big(\mm(a_0,u_1,\la)^{-1}\cdot\mm(u_0,u_1,\la)\big)\cdot\mm(a_0,u_0,\la)=\\
=\Big(\dd{u^j_1}\big(\mm(a_0,u_1,\la)^{-1}\big)\cdot\mm(u_0,u_1,\la)
+\mm(a_0,u_1,\la)^{-1}\cdot\dd{u^j_1}\big(\mm(u_0,u_1,\la)\big)\Big)
\cdot\mm(a_0,u_0,\la)=\\
=\Big(-\mm(a_0,u_1,\la)^{-1}\cdot\dd{u^j_1}\big(\mm(a_0,u_1,\la)\big)\cdot 
\mm(a_0,u_1,\la)^{-1}\cdot\mm(u_0,u_1,\la)+\\
+\mm(a_0,u_1,\la)^{-1}\cdot\dd{u^j_1}\big(\mm(u_0,u_1,\la)\big)\Big)
\cdot\mm(a_0,u_0,\la)=\\
=\mm(a_0,u_1,\la)^{-1}\cdot\Big(-\dd{u^j_1}\big(\mm(a_0,u_1,\la)\big)\cdot
\mm(a_0,u_1,\la)^{-1}+\\
+\dd{u^j_1}\big(\mm(u_0,u_1,\la)\big)\cdot
\mm(u_0,u_1,\la)^{-1}\Big)\cdot\mm(u_0,u_1,\la)\cdot\mm(a_0,u_0,\la)=\\
=\mm(a_0,u_1,\la)^{-1}\cdot Q^j(u_0,u_1,\la)\cdot \mm(u_0,u_1,\la)\cdot\mm(a_0,u_0,\la),
\end{multline}
where 
\begin{gather}
\lb{Lunew1}
Q^j(u_0,u_1,\la)=-\dd{u^j_1}\big(\mm(a_0,u_1,\la)\big)\cdot
\mm(a_0,u_1,\la)^{-1}+\dd{u^j_1}\big(\mm(u_0,u_1,\la)\big)\cdot\mm(u_0,u_1,\la)^{-1}.
\end{gather}
From~\er{Lunew1} and~\er{pdm1} it follows that
\begin{gather}
\lb{pdLnew1}
\dd{u^i_0}\big(Q^j(u_0,u_1,\la)\big)=
\dd{u^i_0}\Big(\dd{u^j_1}
\big(\mm(u_0,u_1,\la)\big)\cdot \mm(u_0,u_1,\la)^{-1}\Big)=0\qquad
\forall\, i=1,\dots,\diunew,\\
\lb{La0new1}
Q^j(a_0,u_1,\la)=0.
\end{gather}
Since, according to Remark~\ref{rratf}, the matrix-function $Q^j(u_0,u_1,\la)$ is rational,
equations~\er{pdLnew1},~\er{La0new1} imply that $Q^j(u_0,u_1,\la)$ is identically zero.
(From~\er{pdLnew1},~\er{La0new1} it follows that the matrix-function $Q^j(u_0,u_1,\la)$ 
is zero on some non-empty open subset of the space with coordinates~$u_0,u_1,\la$.
Since $Q^j(u_0,u_1,\la)$ is rational, this implies that $Q^j(u_0,u_1,\la)$ is identically zero.)

Substituting $Q^j(u_0,u_1,\la)=0$ in~\er{pdtilmn1}, one obtains
$\dd{u^j_1}\big(\hat{\mm}\big)=0$ for all $j=1,\dots,\diunew$.
Therefore, the matrix-function~\er{hatmm1} is of the form $\hat{\mm}=\hat{\mm}(u_0,\la)$. 

Now suppose that we have a MLR $(\mm,\,\lt)$ 
with~$\mm=\mm(u_0,u_1,\la)$ satisfying~\er{pdm1}.
Applying the gauge transformation~\er{guu0} to this MLR, we get the MLR
\begin{gather*}
\hat{\mm}=\cS\big(\mathbf{g}(u_0,\la)\big)\cdot\mm(u_0,u_1,\la)\cdot\mathbf{g}(u_0,\la)^{-1},\qquad
\hat{\lt}=\tdt\big(\mathbf{g}(u_0,\la)\big)\cdot\mathbf{g}(u_0,\la)^{-1}+
\mathbf{g}(u_0,\la)\cdot\lt\cdot\mathbf{g}(u_0,\la)^{-1}.
\end{gather*}
As shown above, $\hat{\mm}$ is of the form $\hat{\mm}=\hat{\mm}(u_0,\la)$.
\end{proof}

\section{Lax representations and Miura-type transformations 
for integrable equations related to the Narita--Itoh--Bogoyavlensky lattice}
\lb{sinb}

In this section we have $\diunew=1$. 
Consider the Narita--Itoh--Bogoyavlensky (NIB) equation~\cite{bogoy88,narita82,itoh87}
\begin{gather}
\lb{be}
\pd_t(\inb)=\inb(\inb_2+\inb_1-\inb_{-1}-\inb_{-2})
\end{gather}
for a scalar function $\inb=\inb(n,t)$.
It is known that equation~\er{be} possesses the MLR
\begin{gather}
\lb{bmlrbc}
\mm^{\text{\tiny{NIB}}}=\begin{pmatrix}
 0& 1& 0\\
 0& 0& 1\\
 -\inb_{0}& 0& \la\\
\end{pmatrix},\qquad
\lt^{\text{\tiny{NIB}}}=\begin{pmatrix}
\inb_{-2} + \inb_{-1}& 0& \la \\
-\la \inb_{0}& \inb_{-1} + \inb_{0}& \la^2\\
-\la^2 \inb_{0}& -\la \inb_{1}& \la^3 + \inb_{0} + \inb_{1}\\
\end{pmatrix}.
\end{gather}

Fix a constant $\mc\in\mathbb{C}$.
Consider the modified NIB equation
\begin{gather}
\label{bpeq}
\pd_t(\pe)=\pe(\pe+\mc)(\pe_2\pe_1 - \pe_{-1}\pe_{-2})
\end{gather}
for a scalar function $\pe=\pe(n,t)$.
It is known (see, e.g.,~\cite{adler-pos14,suris2003}) that the formula
\begin{gather}
\lb{bmtpbc}
\inb=(\pe_2+\mc)\pe_1\pe_0
\end{gather}
determines a MT from~\er{bpeq} to~\er{be}.
Substituting~\er{bmtpbc} in~\er{bmlrbc}, we get the following MLR for equation~\er{bpeq}
\begin{gather}
\lb{bMbce}
\mm(\pe_{0},\pe_{1},\pe_{2},\la)=
\begin{pmatrix}
 0& 1& 0\\
 0& 0& 1\\
 -\pe_0\pe_1(\pe_2+\mc)& 0& \la\\
\end{pmatrix},\\
\lb{bUbce}
\begin{split}
\lt(\pe_{-2}&,\pe_{-1},\pe_{0},\pe_{1},\pe_{2},\pe_{3},\la)=\\
=&\left(\begin{smallmatrix}
 \pe_{-1}\big(\mc(\pe_{-2}+\pe_{0})+\pe_{0} \pe_{-2}+\pe_{0} \pe_{1}\big) & 0 & \la \\
 -\la \pe_{0} \pe_{1} (\pe_{2}+\mc) & 
\pe_{0}\big(\mc(\pe_{-1}+\pe_{1})+\pe_{1}\pe_{-1}+\pe_{1}\pe_{2}\big) & \la^2 \\
 -\la^2\pe_{0}\pe_{1}(\pe_{2}+\mc) & -\la\pe_{1}\pe_{2}(\pe_{3}+\mc) & 
\mc(\pe_{0}\pe_{1}+\pe_{1}\pe_{2})+\pe_{0}\pe_{1}\pe_{2}+\pe_{1}\pe_{2}\pe_{3}+\la^3 \\
\end{smallmatrix}\right).
\end{split}
\end{gather}

The matrix-function~\er{bMbce} satisfies conditions~\er{pdmnew},~\er{pdm01}.
Therefore, we can apply Theorem~\ref{thmuknew} to the MLR~\er{bMbce},~\er{bUbce}
Recall that here $\diunew=1$.
In order to apply Theorem~\ref{thmuknew} in this case, we need to
\begin{itemize}
	\item choose a constant $a_0\in\fik^1$,
	\item substitute $\pe_0=a_0$ in~$\mm(\pe_{0},\pe_{1},\pe_{2},\la)$
given by~\er{bMbce},
	\item consider the matrix-function $\mm(a_0,\pe_0,\pe_1,\la)=\cS^{-1}\big(\mm(a_0,\pe_1,\pe_2,\la)\big)$,
	\item compute $\tilde{\mm}=\tilde{\mm}(\pe_0,\pe_{1},\la)$ given by~\er{tilm},
	\item choose another constant $\tilde{a}_0\in\fik^1$,
	\item substitute $\pe_0=\tilde{a}_0$ in~$\tilde{\mm}(\pe_0,\pe_{1},\la)$ 
	and consider $\tilde{\mm}(\tilde{a}_0,\pe_0,\la)=\cS^{-1}\big(\tilde{\mm}(\tilde{a}_0,\pe_1,\la)\big)$,
	\item compute $\mathbf{g}(\pe_0,\pe_{1},\la)$ given by~\er{guu1},
	\item apply the obtained gauge transformation $\mathbf{g}(\pe_0,\pe_{1},\la)$ 
	to the MLR~\er{bMbce},~\er{bUbce} and compute the gauge equivalent MLR $(\hat{\mm},\,\hat{\lt})$
	given by~\er{hatmm},~\er{hatlt}.
\end{itemize}
It is convenient to take $a_0=1$. 
Using $\mm(a_0,\pe_1,\pe_2,\la)$ obtained by substituting $\pe_0=a_0=1$ in~\er{bMbce},
we get
\begin{gather}
\lb{bmu0u1}
\mm(a_0,\pe_0,\pe_1,\la)=\cS^{-1}\big(\mm(a_0,\pe_1,\pe_2,\la)\big)=
\cS^{-1}\left(\begin{smallmatrix}
 0 & 1 & 0 \\
 0 & 0 & 1 \\
 -\pe_{1} (\pe_{2}+\mc) & 0 & \la \\
\end{smallmatrix}\right)
=\left(\begin{smallmatrix}
 0 & 1 & 0 \\
 0 & 0 & 1 \\
 -\pe_{0} (\pe_{1}+\mc) & 0 & \la \\
\end{smallmatrix}\right).
\end{gather}
Using~\er{bMbce} and~\er{bmu0u1}, 
one computes $\tilde{\mm}=\tilde{\mm}(\pe_0,\pe_{1},\la)$ given by~\er{tilm} as follows
\begin{gather}
\notag
\tilde{\mm}=
\mm(a_0,\pe_1,\pe_2,\la)^{-1}\cdot\mm(\pe_0,\pe_1,\pe_2,\la)\cdot\mm(a_0,\pe_0,\pe_1,\la)=
\left(\begin{smallmatrix}
 0 & \pe_{0} & 0 \\
 0 & 0 & 1 \\
 -\pe_{0} (\pe_{1}+\mc) & 0 & \la \\
\end{smallmatrix}\right).
\end{gather}
Now we take $\tilde{a}_0=1$ and substitute $\pe_0=\tilde{a}_0=1$ in~$\tilde{\mm}(\pe_0,\pe_{1},\la)$.
Then
\begin{gather}
\notag
\tilde{\mm}(\tilde{a}_0,\pe_0,\la)=\cS^{-1}\big(\tilde{\mm}(\tilde{a}_0,\pe_1,\la)\big)=
\cS^{-1}\left(\begin{smallmatrix}
 0 & 1 & 0 \\
 0 & 0 & 1 \\
 -\pe_{1}-\mc & 0 & \la \\
\end{smallmatrix}\right)=
\left(\begin{smallmatrix}
 0 & 1 & 0 \\
 0 & 0 & 1 \\
 -\pe_{0}-\mc & 0 & \la \\
\end{smallmatrix}\right).
\end{gather}
This allows us to compute $\mathbf{g}(\pe_0,\pe_{1},\la)$ given by~\er{guu1} as follows
\begin{gather}
\lb{bcguu1}
\mathbf{g}(\pe_0,\pe_{1},\la)=
\tilde{\mm}(\tilde{a}_0,\pe_0,\la)^{-1}\cdot\mm(a_0,\pe_0,\pe_1,\la)^{-1}=
\left(\begin{smallmatrix}
 \frac{\la}{\pe_{0}+\mc} & -\frac{1}{\pe_{0}+\mc} & 0 \\
 0 & \frac{\la}{\pe_{0} (\pe_{1}+\mc)} & -\frac{1}{\pe_{0} (\pe_{1}+\mc)} \\
 1 & 0 & 0 \\
\end{smallmatrix}\right).
\end{gather}
Finally, applying the obtained gauge 
transformation~\er{bcguu1}	to the MLR~\er{bMbce},~\er{bUbce} 
and computing the corresponding MLR $(\hat{\mm},\,\hat{\lt})$ 
given by~\er{hatmm},~\er{hatlt}, we obtain
\begin{gather}
\lb{bchmm}
\hat{\mm}(\pe_{0},\la)=
\cS\big(\mathbf{g}(\pe_0,\pe_{1},\la)\big)\cdot
\mm(\pe_0,\pe_1,\pe_2,\la)\cdot\mathbf{g}(\pe_0,\pe_{1},\la)^{-1}=
\begin{pmatrix}
 0 & \pe_{0} & 0 \\
 0 & 0 & \pe_{0} \\
 -(\pe_{0}+\mc) & 0 & \la \\
\end{pmatrix},\\
\lb{bchlt}
\begin{split}
&\hat{\lt}(\pe_{-2},\pe_{-1},\pe_{0},\pe_{1},\la)=
\tdt\big(\mathbf{g}(\pe_0,\pe_{1},\la)\big)\cdot
\mathbf{g}(\pe_0,\pe_{1},\la)^{-1}+
\mathbf{g}(\pe_0,\pe_{1},\la)\cdot\lt\cdot\mathbf{g}(\pe_0,\pe_{1},\la)^{-1}=\\
&=\left(
\begin{smallmatrix}
 \pe_{0}(\mc(\pe_{-1}+\pe_{1})+\pe_{-1}(\pe_{-2}+\pe_{1})) & 0 & \la \pe_{-2} \pe_{-1} \\
 -\la\pe_{-1}(\pe_{0}+\mc) & 
\mc(\pe_{-2}\pe_{-1}+\pe_{0}\pe_{1})
+\pe_{0}\pe_{-1}(\pe_{-2}+\pe_{1}) & \la^2\pe_{-1} \\
 -\la^2(\pe_{0}+\mc) & -\la\pe_{0}(\pe_{1}+\mc) & 
\la^3+\pe_{-1}(\mc(\pe_{-2}+\pe_{0})+\pe_{0}(\pe_{-2}+\pe_{1})) \\
\end{smallmatrix}\right).
\end{split}
\end{gather}
In agreement with Theorem~\ref{thmuknew}, we see that
the matrix-function~\er{bchmm} depends only on~$\pe_{0},\,\la$,
in contrast to the matrix-function~\er{bMbce} depending on~$\pe_{0},\,\pe_{1},\,\pe_2,\,\la$.

Fix a constant $\ap\in\mathbb{C}$.
Below we consider the matrices~\er{bchmm},~\er{bchlt} with~$\la$ replaced by~$\ap$.

According to Definition~\ref{dmlpgtnew}, since the matrices~\er{bchmm},~\er{bchlt} 
form a MLR for equation~\er{bpeq}, we can consider the auxiliary linear system
\begin{gather}
\lb{bsyspsi3}
\cS(\Psi)=\hat{\mm}(\pe_{0},\ap)\cdot\Psi,\qquad\quad
\pd_t(\Psi)=\hat{\lt}(\pe_{-2},\pe_{-1},\pe_{0},\pe_{1},\ap)\cdot\Psi,
\end{gather}
which is compatible modulo equation~\eqref{bpeq}. 
Here $\Psi=
\left(\begin{smallmatrix}
\psi^{11}&\psi^{12}&\psi^{13}\\
\psi^{21}&\psi^{22}&\psi^{23}\\
\psi^{31}&\psi^{32}&\psi^{33}
\end{smallmatrix}\right)$
is an invertible $3\times 3$ matrix-function with elements $\psi^{ij}=\psi^{ij}(n,t)$, $\,i,j=1,2,3$.
Using~\er{bchmm},~\er{bchlt} with~$\la$ replaced by~$\ap$,
we see that equations~\er{bsyspsi3} read
\begin{gather}
\lb{bcsp}
\cS(\psi^{1j})=\pe_{0}\psi^{2j},\qquad
\cS(\psi^{2j})=\pe_{0}\psi^{3j},\qquad
\cS(\psi^{3j})=-(\pe_{0}+\mc)\psi^{1j}+\ap\psi^{3j},\qquad j=1,2,3,\\
\lb{dtp}
\pd_t(\psi^{1j})=\pe_{0}(\mc(\pe_{-1}+\pe_{1})+\pe_{-1}(\pe_{-2}+\pe_{1}))\psi^{1j}
+\ap \pe_{-2} \pe_{-1}\psi^{3j},\\
\lb{dtp2}
\pd_t(\psi^{2j})=-\ap\pe_{-1}(\pe_{0}+\mc)\psi^{1j}+
\big(\mc(\pe_{-2}\pe_{-1}+\pe_{0}\pe_{1})
+\pe_{0}\pe_{-1}(\pe_{-2}+\pe_{1})\big)\psi^{2j}+
\ap^2\pe_{-1}\psi^{3j},\\
\lb{dtp3}
\pd_t(\psi^{3j})=-\ap^2(\pe_{0}+\mc)\psi^{1j}
-\ap\pe_{0}(\pe_{1}+\mc)\psi^{2j}+
\big(\ap^3+\pe_{-1}(\mc(\pe_{-2}+\pe_{0})+\pe_{0}(\pe_{-2}+\pe_{1}))\big)\psi^{3j}.
\end{gather}

We set $\vb^1=\dfrac{\psi^{11}}{\psi^{21}}$ and $\vb^2=\dfrac{\psi^{31}}{\psi^{21}}$.
From~\er{bcsp} it follows that 
$\cS(\vv^1)=\dfrac{\cS(\psi^{11})}{\cS(\psi^{21})}=
\dfrac{\pe_{0}\psi^{21}}{\pe_{0}\psi^{31}}=\dfrac{\psi^{21}}{\psi^{31}}=\dfrac{1}{\vb^2}$,
which implies $\vb^2=\dfrac{1}{\cS(\vb^1)}$.
Using this and the relation 
$\dfrac{\psi^{11}}{\psi^{31}}=\dfrac{\psi^{11}/\psi^{21}}{\psi^{31}/\psi^{21}}=
\dfrac{\vb^1}{\vb^2}=\vb^1\cS(\vb^1)$, we get
\begin{gather}
\label{svb2}
\cS(\vv^2)=\frac{\cS(\psi^{31})}{\cS(\psi^{21})}=
\frac{-(\pe_{0}+\mc)\psi^{11}+\ap\psi^{31}}{\pe_{0}\psi^{31}}=
\frac{-(\pe_{0}+\mc)\psi^{11}/\psi^{31}+\ap}{\pe_{0}}=
\frac{-(\pe_{0}+\mc)\vb^1\cS(\vb^1)+\ap}{\pe_{0}}.
\end{gather}
Since $\cS(\vv^2)=\cS\Big(\dfrac{1}{\cS(\vb^1)}\Big)=\dfrac{1}{\cS^2(\vb^1)}$,
equation~\er{svb2} yields
\begin{gather}
\label{bu0vb1}
\pe_{0}=\frac{\ap\cS^2(\vb^1)+\mc}{\vb^1\cS(\vb^1)\cS^2(\vb^1)+1}-\mc.
\end{gather}

We set $\wf=\vb^1$ and $\wf_\ell=\cS^\ell(\vb^1)$ for $\ell\in\zz$. Then~\er{bu0vb1} can be written as
\begin{gather}
\label{bu0vb}
\pe_{0}=\frac{\ap \wf_{2}+\mc}{\wf_{0}\wf_{1}\wf_{2}+1}-\mc.
\end{gather}
For each $\ell\in\zz$, applying the operator~$\cS^\ell$ to equation~\er{bu0vb}, one derives
\begin{gather}
\label{ellbu0vb}
\pe_{\ell}=\frac{\ap \wf_{\ell+2}+\mc}{\wf_{\ell}\wf_{\ell+1}\wf_{\ell+2}+1}-\mc
\qquad\quad\forall\,\ell\in\zz.
\end{gather}
Furthermore, we have
\begin{gather}
\label{psipsi}
\frac{\psi^{11}}{\psi^{21}}=\vb^1=\wf,\qquad\quad
\frac{\psi^{31}}{\psi^{21}}=\vb^2=\frac{1}{\cS(\vb^1)}=\frac{1}{\cS(\wf)}.
\end{gather}

Using~\er{dtp},~\er{dtp2},~\er{psipsi}, one obtains
\begin{multline}
\label{dtwlong}
\pd_t(\wf)=\pd_t\Big(\frac{\psi^{11}}{\psi^{21}}\Big)=
\frac{\pd_t(\psi^{11})}{\psi^{21}}-\frac{\psi^{11}}{\psi^{21}}\cdot\frac{\pd_t(\psi^{21})}{\psi^{21}}
=\pe_{0}(\mc(\pe_{-1}+\pe_{1})+\pe_{-1}(\pe_{-2}+\pe_{1}))\frac{\psi^{11}}{\psi^{21}}
+\ap \pe_{-2} \pe_{-1}\frac{\psi^{31}}{\psi^{21}}+\\
+\frac{\psi^{11}}{\psi^{21}}\cdot
\frac{\ap\pe_{-1}(\pe_{0}+\mc)\psi^{11}}{\psi^{21}}
-\frac{\psi^{11}}{\psi^{21}}\cdot
\frac{\big(\mc(\pe_{-2}\pe_{-1}+\pe_{0}\pe_{1})
+\pe_{0}\pe_{-1}(\pe_{-2}+\pe_{1})\big)\psi^{21}}{\psi^{21}}
-\frac{\psi^{11}}{\psi^{21}}\cdot
\frac{\ap^2\pe_{-1}\psi^{31}}{\psi^{21}}=\\
=\pe_{0}\big(\mc(\pe_{-1}+\pe_{1})+\pe_{-1}(\pe_{-2}+\pe_{1})\big)\wf+
\frac{\ap\pe_{-2}\pe_{-1}}{\cS(\wf)}+\\
+\ap\pe_{-1}(\pe_{0}+\mc)\wf^2
-\big(\mc(\pe_{-2}\pe_{-1}+\pe_{0}\pe_{1})+\pe_{0}\pe_{-1}(\pe_{-2}+\pe_{1})\big)\wf
-\frac{\ap^2\pe_{-1}\wf}{\cS(\wf)}.
\end{multline}
Substituting~\er{ellbu0vb} to the right-hand side of~\er{dtwlong},
we get the following equation
\begin{gather}
\label{dtwf}
\pd_t(\wf)=
\frac{\wf_{0}(\wf_{2} \wf_{1}-\wf_{-1} \wf_{-2}) (\mc \wf_{-1}\wf_{0}-\ap)(\mc\wf_{0}\wf_{1}-\ap) (\ap \wf_{0}+\mc)}{(\wf_{-2} \wf_{-1} \wf_{0}+1) (\wf_{-1} \wf_{0} \wf_{1}+1) (\wf_{0} \wf_{1} \wf_{2}+1)}.
\end{gather}

Recall that system~\er{bsyspsi3} is compatible modulo equation~\eqref{bpeq}
and is equivalent to system~\er{bcsp}, \er{dtp}, \er{dtp2}, \er{dtp3}.
Since equations~\er{bu0vb},~\er{dtwf} are derived from~\er{bcsp},~\er{dtp},~\er{dtp2}, 
system~\er{bu0vb},~\er{dtwf} is compatible modulo equation~\eqref{bpeq} as well.
This implies that \er{bu0vb} is a MT from~\er{dtwf} to~\eqref{bpeq}.

Equation~\er{dtwf} was obtained previously 
in~\cite[equation~(161)]{BIg2016} by a different construction.
In the case when $\mc\neq 0$ and $\ap\neq 0$ equation~\er{dtwf}
can be transformed to~\cite[equation~(5.11)]{xenit2018} by scaling.

As discussed in Remark~\ref{rmtnib}, 
the MT~\er{bu0vb} from~\er{dtwf} to~\eqref{bpeq} 
with arbitrary constants~$\mc,\ap\in\fik$ seems to be new.
In the case when $\mc\neq 0$ and $\ap\neq 0$ the MT~\er{bu0vb} 
can be transformed to the MT in~\cite[page 16]{xenit2018}
by scaling and shift.

\begin{remark}
\lb{ridtwf}
As shown above, the matrices~\er{bchmm},~\er{bchlt} form 
a $\la$-dependent MLR for~\er{bpeq}, 
while formula~\er{bu0vb} is a MT from~\er{dtwf} to~\eqref{bpeq}.
These two facts imply the following.
Substituting~\er{bu0vb} in~\er{bchmm},~\er{bchlt}, one 
obtains a $\la$-dependent MLR for equation~\er{dtwf}.  
This MLR for~\er{dtwf} seems to be new.

Equation~\er{dtwf} is integrable in the sense that 
\er{dtwf} possesses a MLR with parameter~$\la$
and is connected by the MT~\er{bu0vb} 
to the well-known integrable modified NIB equation~\er{bpeq}.
\end{remark}

\section{Integrable equations, Lax representations, 
and Miura-type transformations related to the Zhang--Chen lattice}
\lb{szhch}

In this section we have $\diunew=3$. 
Below we consider scalar functions $\zc^{i}=\zc^{i}(n,t)$, $\,i=1,2,3$,
and the $3$-component vector-functions 
$\zc_\ell=(\zc^1_\ell,\zc^2_\ell,\zc^3_\ell)$, 
where $\zc^{i}_\ell=\cS^\ell(\zc^{i})$ for $\ell\in\zz$ can be regarded as dynamical variables.

D.~Zhang and D.~Chen~\cite{ZhangChen2002} 
introduced the following $3$-component integrable equation
\begin{gather}
\label{zceq}
\left\{
\begin{aligned}
\pd_t(\zc^1)&=\zc^{2}_{1}-\zc^{2}_{0},\\ 
\pd_t(\zc^2)&=\zc^{2}_{0}(\zc^{1}_{-1}-\zc^{1}_{0})+\zc^{3}_{0}-\zc^{3}_{-1},\\
\pd_t(\zc^3)&=\zc^{3}_{0}(\zc^{1}_{-1}-\zc^{1}_{1}).
\end{aligned}\right.
\end{gather}
It is known~\cite{ZhangChen2002,kmw2013} that equation~\er{zceq} possesses the MLR
\begin{gather}
\lb{nmlrbc}
\mm^{\text{\tiny{ZC}}}(\zc_0,\la)=
\begin{pmatrix}
 0 & 1 & 0 \\
 \zc^{2}_{0} & \la+\zc^{1}_{0} & 1 \\
 \zc^{3}_{0} & 0 & 0 \\
\end{pmatrix},\qquad\quad
\lt^{\text{\tiny{ZC}}}(\zc_{-1},\zc_0,\la)=  
\begin{pmatrix}
 -\zc^{1}_{-1} & 1 & 0 \\
 \zc^{2}_{0} & \la & 1 \\
 \zc^{3}_{-1} & 0 & -\zc^{1}_{0} \\
\end{pmatrix}.
\end{gather}

Fix a constant $\mc\in\mathbb{C}$.
In what follows we use the matrices~\er{nmlrbc} with~$\la$ replaced by~$\mc$.

According to Definition~\ref{dmlpgtnew}, since \er{nmlrbc} 
is a MLR for equation~\er{zceq}, we can consider the auxiliary linear system
\begin{gather}
\lb{zsyspsi}
\cS(\Psi)=\mm^{\text{\tiny{ZC}}}(\zc_0,\mc)\cdot\Psi,\qquad\quad
\pd_t(\Psi)=\lt^{\text{\tiny{ZC}}}(\zc_{-1},\zc_0,\mc)\cdot\Psi,
\end{gather}
which is compatible modulo equation~\eqref{zceq}. 
Here $\Psi=
\left(\begin{smallmatrix}
\psi^{11}&\psi^{12}&\psi^{13}\\
\psi^{21}&\psi^{22}&\psi^{23}\\
\psi^{31}&\psi^{32}&\psi^{33}
\end{smallmatrix}\right)$
is an invertible $3\times 3$ matrix-function with elements $\psi^{ij}=\psi^{ij}(n,t)$, $\,i,j=1,2,3$.
Using~\er{nmlrbc} with~$\la$ replaced by~$\mc$, we see that equations~\er{zsyspsi} read
\begin{gather}
\lb{spsic}
\cS(\psi^{1j})=\psi^{2j},\qquad
\cS(\psi^{2j})=\zc^{2}_{0}\psi^{1j}+(\mc+\zc^{1}_{0})\psi^{2j}+\psi^{3j},\qquad
\cS(\psi^{3j})=\zc^{3}_{0}\psi^{1j},\qquad j=1,2,3,\\
\lb{dtpsic}
\pd_t(\psi^{1j})=-\zc^{1}_{-1}\psi^{1j}+\psi^{2j},\qquad
\pd_t(\psi^{2j})=\zc^{2}_{0}\psi^{1j}+\mc\psi^{2j}+\psi^{3j},\qquad
\pd_t(\psi^{3j})=\zc^{3}_{-1}\psi^{1j}-\zc^{1}_{0}\psi^{3j}.
\end{gather}
We set
\begin{gather}
\lb{pezc}
\pe^1=\zc^{1}=\zc^{1}_0,\qquad\quad\pe^2=\frac{\psi^{21}}{\psi^{11}},
\qquad\quad\pe^3=\frac{\psi^{31}}{\psi^{11}},\qquad
\pe^i_\ell=\cS^{\ell}(\pe^i),\qquad i=1,2,3,\quad\ell\in\zz.
\end{gather}
From~\er{spsic},~\er{pezc} it follows that
\begin{gather}
\lb{p21z}
\pe^2_1=\cS(\pe^2)=\frac{\cS(\psi^{21})}{\cS(\psi^{11})}=
\frac{\zc^{2}_{0}\psi^{11}+(\mc+\zc^{1}_{0})\psi^{21}+\psi^{31}}{\psi^{21}}=
\zc^{2}_{0}\frac{1}{\pe^2}+(\mc+\pe^1)+\frac{\pe^3}{\pe^2},\\
\lb{p31z}
\pe^3_1=\cS(\pe^3)=\frac{\cS(\psi^{31})}{\cS(\psi^{11})}=
\frac{\zc^{3}_{0}\psi^{11}}{\psi^{21}}=\zc^{3}_{0}\frac{1}{\pe^2}.
\end{gather}
Since $\zc^i_0=\zc^i$, $\,\pe^i_0=\pe^i$ for $i=1,2,3$ 
and $\pe^1=\zc^{1}$, using~\er{p21z} and~\er{p31z}, we get
\begin{gather}
\lb{nmtpbc}
\left\{
\begin{aligned}
\zc^1&=\pe^{1}_{0},\\ 
\zc^2&=-\mc \pe^{2}_{0}-\pe^{1}_{0} \pe^{2}_{0}+\pe^{2}_{1} \pe^{2}_{0}-\pe^{3}_{0},\\
\zc^3&=\pe^{2}_{0} \pe^{3}_{1}.
\end{aligned}\right.
\end{gather}
Furthermore, for each $\ell\in\zz$, applying the operator~$\cS^\ell$ to equations~\er{nmtpbc},
one obtains
\begin{gather}
\label{ellzzz}
\zc^1_\ell=\pe^{1}_{\ell},\qquad\quad
\zc^2_\ell=-\mc \pe^{2}_{\ell}-\pe^{1}_{\ell} \pe^{2}_{\ell}+\pe^{2}_{\ell+1} \pe^{2}_{\ell}-\pe^{3}_{\ell},
\qquad\quad
\zc^3_\ell=\pe^{2}_{\ell} \pe^{3}_{\ell+1}\qquad\quad\forall\,\ell\in\zz.
\end{gather}
Using \er{zceq}, \er{dtpsic}, \er{pezc}, \er{ellzzz}, we derive
\begin{gather}
\label{dpe1}
\pd_t(\pe^1)=\pd_t(\zc^{1})=\zc^{2}_{1}-\zc^{2}_{0}=
\mc\pe^{2}_{0}-\mc \pe^{2}_{1}+\pe^{1}_{0} \pe^{2}_{0}-\pe^{2}_{1} \pe^{2}_{0}
-\pe^{1}_{1} \pe^{2}_{1}+\pe^{2}_{1} \pe^{2}_{2}+\pe^{3}_{0}-\pe^{3}_{1},\\
\label{dpe2}
\begin{split}
\pd_t(\pe^2)&=\pd_t\Big(\frac{\psi^{21}}{\psi^{11}}\Big)=
\frac{\pd_t(\psi^{21})}{\psi^{11}}-\frac{\psi^{21}}{\psi^{11}}\cdot\frac{\pd_t(\psi^{11})}{\psi^{11}}=
\frac{\zc^{2}_{0}\psi^{11}+\mc\psi^{21}+\psi^{31}}{\psi^{11}}
-\pe^2\cdot\frac{-\zc^{1}_{-1}\psi^{11}+\psi^{21}}{\psi^{11}}=\\
&=\zc^{2}_{0}+\mc\pe^2+\pe^3-\pe^2(-\zc^{1}_{-1}+\pe^2)=
\pe^{2}_{0}(\pe^{1}_{-1}-\pe^{1}_{0}-\pe^{2}_{0}+\pe^{2}_{1}),
\end{split}\\
\label{dpe3}
\begin{split}
\pd_t(\pe^3)&=\pd_t\Big(\frac{\psi^{31}}{\psi^{11}}\Big)=
\frac{\pd_t(\psi^{31})}{\psi^{11}}-\frac{\psi^{31}}{\psi^{11}}\cdot\frac{\pd_t(\psi^{11})}{\psi^{11}}=
\frac{\zc^{3}_{-1}\psi^{11}-\zc^{1}_{0}\psi^{31}}{\psi^{11}}
-\pe^3\cdot\frac{-\zc^{1}_{-1}\psi^{11}+\psi^{21}}{\psi^{11}}=\\
&=\zc^{3}_{-1}-\zc^{1}_{0}\pe^3-\pe^3(-\zc^{1}_{-1}+\pe^2)
=\pe^{3}_{0}(\pe^{1}_{-1}-\pe^{1}_{0}+\pe^{2}_{-1}-\pe^{2}_{0}).
\end{split}
\end{gather}
From~\er{dpe1},~\er{dpe2},~\er{dpe3}, one gets the $3$-component equation
\begin{gather}
\label{npeq}
\left\{
\begin{aligned}
\pd_t(\pe^1)&=\mc \pe^{2}_{0}-\mc \pe^{2}_{1}+\pe^{1}_{0} \pe^{2}_{0}
-\pe^{2}_{1} \pe^{2}_{0}-\pe^{1}_{1} \pe^{2}_{1}+\pe^{2}_{1} \pe^{2}_{2}+\pe^{3}_{0}-\pe^{3}_{1},\\ 
\pd_t(\pe^2)&=\pe^{2}_{0}(\pe^{1}_{-1}-\pe^{1}_{0}-\pe^{2}_{0}+\pe^{2}_{1}),\\
\pd_t(\pe^3)&=\pe^{3}_{0}(\pe^{1}_{-1}-\pe^{1}_{0}+\pe^{2}_{-1}-\pe^{2}_{0}).
\end{aligned}\right.
\end{gather}

Recall that system~\er{zsyspsi} is compatible modulo equation~\eqref{zceq}
and is equivalent to system~\er{spsic},~\er{dtpsic}. 
Since equations~\er{nmtpbc},~\er{npeq} 
are derived from~\er{spsic},~\er{dtpsic} (by means of introducing~\er{pezc}), 
system~\er{nmtpbc},~\er{npeq} is compatible modulo equation~\eqref{zceq} as well.
This implies that \er{nmtpbc} is a MT from~\er{npeq} to~\er{zceq}.

Below we use the notation~\er{uldiunew} with $\diunew=3$.
That is, for each $\ell\in\zz$ one has $\pe_\ell=(\pe^1_\ell,\pe^2_\ell,\pe^3_\ell)$.

Substituting~\er{nmtpbc} in~\er{nmlrbc}, we obtain the following MLR for equation~\er{npeq}
\begin{gather}
\lb{Mbce}
\mm(\pe_{0},\pe_{1},\la)=
\begin{pmatrix}
 0 & 1 & 0 \\
 -\mc \pe^{2}_{0}-\pe^{1}_{0} \pe^{2}_{0}+\pe^{2}_{1} \pe^{2}_{0}-\pe^{3}_{0} & \la+\pe^{1}_{0} & 1 \\
 \pe^{2}_{0} \pe^{3}_{1} & 0 & 0 \\
\end{pmatrix},\\
\lb{Ubce}
\lt(\pe_{-1},\pe_{0},\pe_{1},\la)=
\begin{pmatrix}
 -\pe^{1}_{-1} & 1 & 0 \\
 -\mc \pe^{2}_{0}-\pe^{1}_{0} \pe^{2}_{0}+\pe^{2}_{1} \pe^{2}_{0}-\pe^{3}_{0} & \la & 1 \\
 \pe^{2}_{-1} \pe^{3}_{0} & 0 & -\pe^{1}_{0} \\
\end{pmatrix}.
\end{gather}

The matrix~\er{Mbce} satisfies condition~\er{pdm1}.
Therefore, we can apply Theorem~\ref{thmuk1} to the MLR~\er{Mbce}, \er{Ubce}.
In order to apply Theorem~\ref{thmuk1} in this case, we need to 
do the following.
\begin{itemize}
	\item Choose a constant vector $a_0=(a^1_0,a^2_0,a^3_0)\in\fik^3$.
	\item Substitute $\pe_0=a_0$ in~$\mm(\pe_{0},\pe_{1},\la)$
given by~\er{Mbce} (i.e., substitute $u^i_0=a^i_0$ for $i=1,2,3$ in~\er{Mbce}).
	\item Compute $\mathbf{g}(\pe_0,\la)$ given by formula~\er{guu0}, 
	which involves the matrix 
	\begin{gather}
\lb{inMmbe}
\mm(a_0,u_1,\la)^{-1}=
\left(\begin{smallmatrix}
 0 & 1 & 0 \\
 -\mc a^{2}_{0}-a^{1}_{0} a^{2}_{0}+\pe^{2}_{1} a^{2}_{0}-a^{3}_{0} & \la+a^{1}_{0} & 1 \\
 a^{2}_{0} \pe^{3}_{1} & 0 & 0 \\
\end{smallmatrix}\right)^{-1}=
\left(\begin{smallmatrix}
 0 & 0 & \frac{1}{a^{2}_{0} \pe^{3}_{1}} \\
 1 & 0 & 0 \\
 -\la-a^{1}_{0} & 1 & 
\frac{\mc a^{2}_{0}+a^{1}_{0} a^{2}_{0}-\pe^{2}_{1} a^{2}_{0}+a^{3}_{0}}{a^{2}_{0} \pe^{3}_{1}} \\
\end{smallmatrix}\right).
\end{gather}
\item Apply the obtained gauge transformation $\mathbf{g}(\pe_0,\la)$ 
to the MLR~\er{Mbce},~\er{Ubce} and compute the gauge equivalent MLR $(\hat{\mm},\,\hat{\lt})$
given by~\er{hatmm1},~\er{hatlt1}.
\end{itemize}
We can choose for $a_0=(a^1_0,a^2_0,a^3_0)$ any constant vector 
such that \er{inMmbe} is well defined.
In order to make formula~\er{inMmbe} as simple as possible, we take 
\begin{gather}
\lb{aaa} 
a^1_0=0,\qquad a^2_0=1,\qquad a^3_0=0.
\end{gather}
Then, using~\er{inMmbe} with~\er{aaa}, we see that \er{guu0} reads
\begin{gather}
\lb{guu0z}
\mathbf{g}(u_0,\la)=
\cS^{-1}\big(\mm(a_0,u_1,\la)^{-1}\big)=
\cS^{-1}\left(\begin{smallmatrix}
 0 & 0 & \frac{1}{\pe^{3}_{1}} \\
 1 & 0 & 0 \\
 -\la & 1 & \frac{\mc-\pe^{2}_{1}}{\pe^{3}_{1}} \\
\end{smallmatrix}\right)=
\left(\begin{smallmatrix}
 0 & 0 & \frac{1}{\pe^{3}_{0}} \\
 1 & 0 & 0 \\
 -\la & 1 & \frac{\mc-\pe^{2}_{0}}{\pe^{3}_{0}} \\
\end{smallmatrix}\right).
\end{gather}

Applying the obtained gauge transformation~\er{guu0z} to the MLR~\er{Mbce},~\er{Ubce} 
and computing the MLR $(\hat{\mm},\,\hat{\lt})$ given by~\er{hatmm1},~\er{hatlt1}, we get
\begin{gather}
\lb{zchmm}
\hat{\mm}(\pe_{0},\la)=
\cS\big(\mathbf{g}(\pe_0,\la)\big)\cdot\mm(\pe_0,\pe_1,\la)\cdot\mathbf{g}(\pe_0,\la)^{-1}=
\begin{pmatrix}
 0 & \pe^{2}_{0} & 0 \\
 \pe^{2}_{0}-\mc & \la & 1 \\
 \pe^{3}_{0}-\mc \pe^{1}_{0}+\pe^{2}_{0} \pe^{1}_{0} & \la \pe^{1}_{0}-\pe^{2}_{0} \pe^{1}_{0}-\pe^{3}_{0} & \pe^{1}_{0} \\
\end{pmatrix},\\
\lb{zchlt}
\begin{split}
\hat{\lt}(&\pe_{-1},\pe_{0},\pe_{1},\la)=\tdt\big(\mathbf{g}(\pe_0,\la)\big)\cdot
\mathbf{g}(\pe_0,\la)^{-1}+
\mathbf{g}(\pe_0,\la)\cdot\lt\cdot\mathbf{g}(\pe_0,\la)^{-1}=\\
&=
\left(\begin{smallmatrix}
 \pe^{2}_{0}-\pe^{1}_{-1}-\pe^{2}_{-1} & \pe^{2}_{-1} & 0 \\
 \pe^{2}_{0}-\mc & \la-\pe^{1}_{-1} & 1 \\
 \pe^{3}_{0}-\mc \pe^{1}_{-1}-\mc \pe^{2}_{-1}+\mc \pe^{2}_{0}+\pe^{1}_{0} \pe^{2}_{0}+
\pe^{2}_{-1} \pe^{2}_{0}-\pe^{2}_{0} \pe^{2}_{1} & \la \pe^{1}_{-1}+
\mc \pe^{2}_{-1}-\mc \pe^{2}_{0}-\pe^{1}_{0} \pe^{2}_{0}-\pe^{2}_{-1} \pe^{2}_{0}+
\pe^{2}_{0} \pe^{2}_{1}-\pe^{3}_{0} & 0 \\
\end{smallmatrix}\right).
\end{split}
\end{gather}
In agreement with Theorem~\ref{thmuk1}, we see that
the matrix~\er{zchmm} depends only on~$\pe_{0},\,\la$,
in contrast to the matrix~\er{Mbce} depending on~$\pe_{0},\,\pe_{1},\,\la$.

Equation~\er{npeq} is integrable, 
since \er{npeq} has the MLR~\er{zchmm},~\er{zchlt} with parameter~$\la$
and is connected by the MT~\er{nmtpbc} 
to the integrable Zhang--Chen equation~\er{zceq}.
Equation~\er{npeq}, the MLR~\er{zchmm},~\er{zchlt}, 
and the MT~\er{nmtpbc} seem to be new.

Fix a constant $\ap\in\mathbb{C}$.
In what follows we use the matrices~\er{zchmm},~\er{zchlt} with~$\la$ replaced by~$\ap$.

Since \er{zchmm},~\er{zchlt} 
form a MLR for equation~\er{npeq}, we can consider the auxiliary linear system
\begin{gather}
\lb{hzsyspsi}
\cS(\hat{\Psi})=\hat{\mm}(\pe_{0},\ap)\cdot\hat{\Psi},\qquad\quad
\pd_t(\hat{\Psi})=\hat{\lt}(\pe_{-1},\pe_{0},\pe_{1},\ap)\cdot\hat{\Psi},
\end{gather}
which is compatible modulo equation~\eqref{npeq}. 
Here $\hat{\Psi}=
\left(\begin{smallmatrix}
\hat{\psi}^{11}&\hat{\psi}^{12}&\hat{\psi}^{13}\\
\hat{\psi}^{21}&\hat{\psi}^{22}&\hat{\psi}^{23}\\
\hat{\psi}^{31}&\hat{\psi}^{32}&\hat{\psi}^{33}
\end{smallmatrix}\right)$
is an invertible $3\times 3$ matrix-function with elements $\hat{\psi}^{ij}=\hat{\psi}^{ij}(n,t)$, $\,i,j=1,2,3$.
Using~\er{zchmm},~\er{zchlt} with~$\la$ replaced by~$\ap$, we see that equations~\er{hzsyspsi} read
\begin{gather}
\begin{gathered}
\lb{hspsic}
\cS(\hat{\psi}^{1j})=\pe^{2}_{0}\hat{\psi}^{2j},\qquad
\cS(\hat{\psi}^{2j})=(\pe^{2}_{0}-\mc)\hat{\psi}^{1j}+\ap\hat{\psi}^{2j}+\hat{\psi}^{3j},\\
\cS(\hat{\psi}^{3j})=(\pe^{3}_{0}-\mc \pe^{1}_{0}+\pe^{2}_{0} \pe^{1}_{0})\hat{\psi}^{1j}+
(\ap\pe^{1}_{0}-\pe^{2}_{0} \pe^{1}_{0}-\pe^{3}_{0})\hat{\psi}^{2j}+\pe^{1}_{0}\hat{\psi}^{3j},
\qquad j=1,2,3,
\end{gathered}\\
\lb{hdtpsic}
\begin{gathered}
\pd_t(\hat{\psi}^{1j})=(\pe^{2}_{0}-\pe^{1}_{-1}-\pe^{2}_{-1})\hat{\psi}^{1j}+\pe^{2}_{-1}\hat{\psi}^{2j},
\qquad
\pd_t(\hat{\psi}^{2j})=(\pe^{2}_{0}-\mc)\hat{\psi}^{1j}+(\ap-\pe^{1}_{-1})\hat{\psi}^{2j}+\hat{\psi}^{3j},\\
\pd_t(\hat{\psi}^{3j})=
(\pe^{3}_{0}-\mc \pe^{1}_{-1}-\mc \pe^{2}_{-1}+\mc \pe^{2}_{0}
+\pe^{1}_{0} \pe^{2}_{0}+\pe^{2}_{-1} \pe^{2}_{0}-\pe^{2}_{0} \pe^{2}_{1})\hat{\psi}^{1j}+\\
+(\ap\pe^{1}_{-1}+\mc\pe^{2}_{-1}-\mc \pe^{2}_{0}-\pe^{1}_{0}\pe^{2}_{0}-\pe^{2}_{-1}\pe^{2}_{0}+
\pe^{2}_{0} \pe^{2}_{1}-\pe^{3}_{0})\hat{\psi}^{2j}.
\end{gathered}
\end{gather}
We set
\begin{gather}
\lb{hpezc}
\vv^1=\pe^{1}=\pe^{1}_0,\qquad\quad\vv^2=\frac{\hat{\psi}^{21}}{\hat{\psi}^{11}},
\qquad\quad\vv^3=\frac{\hat{\psi}^{31}}{\hat{\psi}^{11}},\qquad
\vv^i_\ell=\cS^{\ell}(\vv^i),\qquad i=1,2,3,\quad\ell\in\zz.
\end{gather}
From~\er{hspsic},~\er{hpezc} it follows that
\begin{gather}
\lb{hp21z}
\vv^2_1=\cS(\vv^2)=\frac{\cS(\hat{\psi}^{21})}{\cS(\hat{\psi}^{11})}=
\frac{(\pe^{2}_{0}-\mc)\hat{\psi}^{11}+\ap\hat{\psi}^{21}+\hat{\psi}^{31}}{\pe^{2}_{0}\hat{\psi}^{21}}=
\frac{(\pe^{2}_{0}-\mc)}{\pe^{2}_{0}\vv^2}+\frac{\ap}{\pe^{2}_{0}}+\frac{\vv^3}{\pe^{2}_{0}\vv^2},\\
\begin{gathered}
\lb{hp31z}
\vv^3_1=\cS(\vv^3)=\frac{\cS(\hat{\psi}^{31})}{\cS(\hat{\psi}^{11})}=
\frac{(\pe^{3}_{0}-\mc \pe^{1}_{0}+\pe^{2}_{0} \pe^{1}_{0})\hat{\psi}^{11}+
(\ap\pe^{1}_{0}-\pe^{2}_{0} \pe^{1}_{0}-\pe^{3}_{0})\hat{\psi}^{21}+
\pe^{1}_{0}\hat{\psi}^{31}}{\pe^{2}_{0}\hat{\psi}^{21}}=\\
=\frac{(\pe^{3}_{0}-\mc \pe^{1}_{0}+\pe^{2}_{0} \pe^{1}_{0})}{\pe^{2}_{0}\vv^2}+
\frac{(\ap\pe^{1}_{0}-\pe^{2}_{0} \pe^{1}_{0}-\pe^{3}_{0})}{\pe^{2}_{0}}+
\frac{\pe^{1}_{0}\vv^3}{\pe^{2}_{0}\vv^2}.
\end{gathered}
\end{gather}
Using the equalities 
$\vv^1=\pe^{1}$, $\,\pe^i_0=\pe^i$, $\,\vv^i_0=\vv^i$, $\,i=1,2,3$, 
and~\er{hp21z},~\er{hp31z}, we get
\begin{gather}
\lb{hnmtpbc}
\left\{
\begin{aligned}
\pe^1&=\vv^{1}_{0},\\ 
\pe^2&=\frac{\ap \vv^{2}_{0}+\vv^{3}_{0}-\mc}{\vv^{2}_{0} \vv^{2}_{1}-1},\\
\pe^3&=\frac{\vv^{2}_{0} (\vv^{2}_{1} \vv^{1}_{0}-\vv^{1}_{0}-\vv^{3}_{1}) (\ap \vv^{2}_{0}+\vv^{3}_{0}-\mc)}{(\vv^{2}_{0}-1) (\vv^{2}_{0} \vv^{2}_{1}-1)}.
\end{aligned}\right.
\end{gather}

For each $\ell\in\zz$, applying the operator~$\cS^\ell$ to equations~\er{hnmtpbc}, one obtains
\begin{gather}
\lb{hpeell}
\pe^1_\ell=\vv^{1}_{\ell},\qquad
\pe^2_\ell=\frac{\ap \vv^{2}_{\ell}+\vv^{3}_{\ell}-\mc}{\vv^{2}_{\ell} \vv^{2}_{\ell+1}-1},\qquad
\pe^3_\ell=\frac{\vv^{2}_{\ell} (\vv^{2}_{\ell+1} \vv^{1}_{\ell}-\vv^{1}_{\ell}-\vv^{3}_{\ell+1}) (\ap \vv^{2}_{\ell}+\vv^{3}_{\ell}-\mc)}{(\vv^{2}_{\ell}-1) (\vv^{2}_{\ell} \vv^{2}_{\ell+1}-1)}\quad\qquad\forall\,\ell\in\zz.
\end{gather}
Using \er{npeq}, \er{hdtpsic}, \er{hpezc}, we derive
\begin{gather}
\lb{dtv1p}
\pd_t(\vv^1)=\pd_t(\pe^{1})=\mc \pe^{2}_{0}-\mc \pe^{2}_{1}+\pe^{1}_{0} \pe^{2}_{0}
-\pe^{2}_{1} \pe^{2}_{0}-\pe^{1}_{1} \pe^{2}_{1}+\pe^{2}_{1} \pe^{2}_{2}+\pe^{3}_{0}-\pe^{3}_{1},\\ 
\lb{dtv2p}
\begin{split}
\pd_t(\vv^2)&=\pd_t\Big(\frac{\hat{\psi}^{21}}{\hat{\psi}^{11}}\Big)=
\frac{\pd_t(\hat{\psi}^{21})}{\hat{\psi}^{11}}
-\frac{\hat{\psi}^{21}}{\hat{\psi}^{11}}\cdot\frac{\pd_t(\hat{\psi}^{11})}{\hat{\psi}^{11}}=\\
&=\frac{(\pe^{2}_{0}-\mc)\hat{\psi}^{11}+(\ap-\pe^{1}_{-1})\hat{\psi}^{21}+\hat{\psi}^{31}}{\hat{\psi}^{11}}
-\frac{\hat{\psi}^{21}}{\hat{\psi}^{11}}\cdot\frac{(\pe^{2}_{0}-\pe^{1}_{-1}-\pe^{2}_{-1})\hat{\psi}^{11}+\pe^{2}_{-1}\hat{\psi}^{21}}{\hat{\psi}^{11}}=\\
&=(\pe^{2}_{0}-\mc)+(\ap-\pe^{1}_{-1})\vv^2+\vv^3
-\vv^2\big((\pe^{2}_{0}-\pe^{1}_{-1}-\pe^{2}_{-1})+\pe^{2}_{-1}\vv^2\big).
\end{split}\\
\lb{dtv3p}
\begin{gathered}
\pd_t(\vv^3)=\pd_t\Big(\frac{\hat{\psi}^{31}}{\hat{\psi}^{11}}\Big)=
\frac{\pd_t(\hat{\psi}^{31})}{\hat{\psi}^{11}}-\frac{\hat{\psi}^{31}}{\hat{\psi}^{11}}\cdot\frac{\pd_t(\hat{\psi}^{11})}{\hat{\psi}^{11}}=
\frac{(\pe^{3}_{0}-\mc \pe^{1}_{-1}-\mc \pe^{2}_{-1}+\mc \pe^{2}_{0}+\pe^{1}_{0} \pe^{2}_{0}
+\pe^{2}_{-1} \pe^{2}_{0}-\pe^{2}_{0} \pe^{2}_{1})\hat{\psi}^{11}}{\hat{\psi}^{11}}+\\
+\frac{(\ap\pe^{1}_{-1}+\mc\pe^{2}_{-1}-\mc \pe^{2}_{0}-\pe^{1}_{0}\pe^{2}_{0}-\pe^{2}_{-1}\pe^{2}_{0}+
\pe^{2}_{0} \pe^{2}_{1}-\pe^{3}_{0})\hat{\psi}^{21}}{\hat{\psi}^{11}}
-\frac{\hat{\psi}^{31}}{\hat{\psi}^{11}}\cdot\frac{(\pe^{2}_{0}-\pe^{1}_{-1}-\pe^{2}_{-1})\hat{\psi}^{11}+\pe^{2}_{-1}\hat{\psi}^{21}}{\hat{\psi}^{11}}=\\
=(\pe^{3}_{0}-\mc \pe^{1}_{-1}-\mc \pe^{2}_{-1}+\mc \pe^{2}_{0}+\pe^{1}_{0} \pe^{2}_{0}
+\pe^{2}_{-1} \pe^{2}_{0}-\pe^{2}_{0} \pe^{2}_{1})+\\
+(\ap\pe^{1}_{-1}+\mc\pe^{2}_{-1}-\mc \pe^{2}_{0}-\pe^{1}_{0}\pe^{2}_{0}-\pe^{2}_{-1}\pe^{2}_{0}+
\pe^{2}_{0} \pe^{2}_{1}-\pe^{3}_{0})\vv^2
-\vv^3\big((\pe^{2}_{0}-\pe^{1}_{-1}-\pe^{2}_{-1})+\pe^{2}_{-1}\vv^2\big).
\end{gathered}
\end{gather}
Substituting~\er{hpeell} in~\er{dtv1p},~\er{dtv2p},~\er{dtv3p}, 
we get a $3$-component equation of the form
\begin{gather}
\label{hnpeq}
\left\{
\begin{aligned}
\pd_t(\vv^1)&=F^1(\mc,\ap,\vv^1_\ell,\vv^2_\ell,\vv^3_\ell,\dots),\\ 
\pd_t(\vv^2)&=F^2(\mc,\ap,\vv^1_\ell,\vv^2_\ell,\vv^3_\ell,\dots),\\
\pd_t(\vv^3)&=F^3(\mc,\ap,\vv^1_\ell,\vv^2_\ell,\vv^3_\ell,\dots),
\end{aligned}\right.
\end{gather}
where the functions $F^i(\mc,\ap,\vv^1_\ell,\vv^2_\ell,\vv^3_\ell,\dots)$, $\,i=1,2,3$, 
are obtained by substituting~\er{hpeell} to the right-hand sides
of equations~\er{dtv1p},~\er{dtv2p},~\er{dtv3p}.
We do not present explicit formulas for these functions,
since they are rather cumbersome.

Recall that system~\er{hzsyspsi} is compatible modulo equation~\eqref{npeq}
and is equivalent to system~\er{hspsic},~\er{hdtpsic}. 
Since equations~\er{hnmtpbc}--\er{hnpeq} 
are derived from~\er{hspsic},~\er{hdtpsic} (by means of introducing~\er{hpezc}), 
system~\er{hnmtpbc}--\er{hnpeq} is compatible modulo equation~\eqref{npeq} as well.
This implies that \er{hnmtpbc} is a MT from~\er{hnpeq} to~\er{npeq}.

Computing the composition of the MTs~\er{nmtpbc} and~\er{hnmtpbc}, 
we get the following MT
\begin{gather}
\lb{cmtzc}
\left\{
\begin{aligned}
\zc^1&=\vv^{1}_{0},\\ 
\zc^2&=\frac{(\ap \vv^{2}_{0}+\vv^{3}_{0}-\mc)
\big(\vv^{2}_{1}(\ap \vv^{2}_{0}-\ap+\mc \vv^{2}_{2}-\mc \vv^{2}_{0} \vv^{2}_{2}-
\vv^{1}_{0} \vv^{2}_{0} \vv^{2}_{2} \vv^{2}_{1}+\vv^{1}_{0} \vv^{2}_{0}+\vv^{1}_{0} \vv^{2}_{2}
+\vv^{2}_{0} \vv^{2}_{2} \vv^{3}_{1})-\vv^{1}_{0}-\vv^{3}_{1}\big)}{(\vv^{2}_{0}-1) (\vv^{2}_{0} \vv^{2}_{1}-1) (\vv^{2}_{1} \vv^{2}_{2}-1)},\\
\zc^3&=\frac{\vv^{2}_{1}(\vv^{2}_{2} \vv^{1}_{1}-\vv^{1}_{1}-\vv^{3}_{2})(\ap \vv^{2}_{0}+\vv^{3}_{0}-\mc)^2}{(\vv^{2}_{1}-1) (\vv^{2}_{0}\vv^{2}_{1}-1) (\vv^{2}_{1}\vv^{2}_{2}-1)}
\end{aligned}\right.
\end{gather}
from equation~\er{hnpeq} to equation~\er{zceq}.

\begin{remark}
\lb{rihnpeq}
As shown above, the matrices~\er{zchmm},~\er{zchlt} form 
a $\la$-dependent MLR for equation~\er{npeq}, 
while \er{hnmtpbc} is a MT from~\er{hnpeq} to~\er{npeq}.
These two facts yield the following.
Substituting~\er{hnmtpbc} in~\er{zchmm},~\er{zchlt}, 
we get a $\la$-dependent MLR for equation~\er{hnpeq}.  
This MLR for~\er{hnpeq} seems to be new.

Equation~\er{hnpeq} is integrable in the sense that 
\er{hnpeq} has a MLR with parameter~$\la$
and is connected by the MT~\er{cmtzc} 
to the integrable Zhang--Chen equation~\er{zceq}.
Equation~\er{hnpeq} and the MTs~\er{hnmtpbc},~\er{cmtzc} seem to be new as well.
\end{remark}

\section*{Acknowledgments}

The author would like to thank Pavlos Xenitidis for useful discussions.

This work was supported by the Russian Science Foundation (grant No. 25-21-00454, 
\url{https://rscf.ru/en/project/25-21-00454/} ).

\end{document}